\documentclass[lettersize,journal]{IEEEtran}
\IEEEoverridecommandlockouts
\usepackage[nospace,noadjust]{cite}
\usepackage{amsmath,amssymb,amsfonts,times}
\usepackage{graphicx}
\usepackage{textcomp}
\usepackage{xcolor}
\usepackage{subfigure}
\usepackage{epsfig}
\usepackage{multirow}
\usepackage{algorithm}
\usepackage{algorithmicx}
\usepackage{algpseudocode}
\usepackage{array}
\usepackage{epstopdf}
\usepackage{color}
\usepackage{diagbox}
\usepackage{bm}
\usepackage{tcolorbox}
\usepackage{pdfsync}
\usepackage{bbm}

\usepackage{url}
\usepackage[hidelinks]{hyperref}
\usepackage{amsmath,amsthm,amssymb,amsbsy}
\usepackage{paralist}
\usepackage{xcolor}
\usepackage{color}
\usepackage{comment}
\usepackage{multirow}

\usepackage{fancyhdr}
\usepackage{cite}
\usepackage{cleveref}

\newtheorem{lemma}{Lemma}
\newtheorem{cor}{Corollary}
\newtheorem{defi}{Definition}

\newtheorem{theorem}{Theorem}

\theoremstyle{remark}

\theoremstyle{problem}

\newcommand{\R}{\mathbb{R}}
\def \real    { \mathbb{R} }

\newcommand{\C}{\mathbb{C}}

\newcommand{\e}{\begin{equation}}
\newcommand{\ee}{\end{equation}}
\newcommand{\en}{\begin{equation*}}
\newcommand{\een}{\end{equation*}}
\newcommand{\eqn}{\begin{eqnarray}}
\newcommand{\eeqn}{\end{eqnarray}}
\newcommand{\bmat}{\begin{bmatrix}}
\newcommand{\emat}{\end{bmatrix}}

\DeclareMathAlphabet\mathbfcal{OMS}{cmsy}{b}{n}
\renewcommand{\P}[1]{\operatorname{\mathbb{P}}\left(#1\right)}
\newcommand{\E}{\operatorname{\mathbb{E}}}

\newcommand{\vct}[1]{\boldsymbol{#1}}
\newcommand{\mtx}[1]{\boldsymbol{#1}}

\newcommand{\<}{\langle}
\renewcommand{\>}{\rangle}

\newcommand{\trace}{\operatorname{trace}}

\newcommand{\set}[1]{\mathbb{#1}}

\DeclareMathOperator*{\argmin}{\text{arg~min}}

\newcommand{\wh}{\widehat}

\newcommand{\wt}{\widetilde}
\newcommand{\ol}{\overline}
\newcommand{\nqbit}{n}

\newcommand{\innerprod}[2]{\left\langle #1,  #2 \right\rangle}

\newcommand{\calA}{\mathcal{A}}

\newcommand{\calC}{\mathcal{C}}

\newcommand{\calF}{\mathcal{F}}

\newcommand{\calN}{\mathcal{N}}

\newcommand{\calP}{\mathcal{P}}

\newcommand{\calX}{\mathcal{X}}

\newcommand{\va}{\vct{a}}

\newcommand{\vp}{\vct{p}}

\newcommand{\vr}{\vct{r}}

\newcommand{\vu}{\vct{u}}

\newcommand{\vw}{\vct{w}}
\newcommand{\vx}{\vct{x}}

\newcommand{\vtheta}{\vct{\theta}}

\newcommand{\vdelta}{\vct{\delta}}
\newcommand{\veta}{\vct{\eta}}

\newcommand{\vrho}{\vct{\rho}}
\newcommand{\vzero}{\vct{0}}

\newcommand{\mA}{\mtx{A}}
\newcommand{\mB}{\mtx{B}}

\newcommand{\mD}{\mtx{D}}

\newcommand{\mH}{\mtx{H}}

\newcommand{\mN}{\mtx{N}}

\newcommand{\mS}{\mtx{S}}

\newcommand{\mW}{\mtx{W}}
\newcommand{\mX}{\mtx{X}}
\newcommand{\mY}{\mtx{Y}}

\newcommand{\mDelta}{\mtx{\Delta}}

\newcommand{\mSigma}{\mtx{\Sigma}}

\newcommand{\mId}{{\bf I}}

\newcommand{\setN}{\set{N}}

\newcommand{\setS}{\set{S}}
\newcommand{\setT}{\set{T}}

\newcommand{\setX}{\set{X}}

\graphicspath{{./figs/}}

\newlength{\imgwidth}
\newboolean{twoColVersion}
\setboolean{twoColVersion}{false}
\newcommand{\twoCol}[2]{\ifthenelse{\boolean{twoColVersion}} {#1} {#2} }

\newcommand{\Q}{\operatorname{\mathbb{Q}}}

\makeatletter
\def\@IEEEsectpunct{\ \,}
\def\paragraph{\@startsection{paragraph}{4}{\z@}{1.5ex plus 1.5ex minus 0.5ex}%
{0ex}{\normalfont\normalsize\sffamily\bfseries}}
\makeatother

\title{Sample-Efficient Quantum State Tomography for Structured Quantum States in One Dimension}

\author{Zhen~Qin,~\IEEEmembership{Member,~IEEE,} Casey~Jameson, Alireza Goldar, Michael~B.~Wakin,~\IEEEmembership{Fellow,~IEEE,} \\ Zhexuan~Gong, and~Zhihui~Zhu,~\IEEEmembership{Member,~IEEE}
\thanks{Zhen Qin is with the Michigan Institute for Computational Discovery and Engineering, Department of Electrical Engineering and Computer Science, Department of Statistic, University of Michigan, Ann Arbor, MI 48109 USA, and also with the Department of Computer Science and Engineering, the Ohio State University, Columbus OH 43210 USA. (e-mail: zhenqin@umich.edu).}
\thanks{Zhihui Zhu is with the Department of Computer Science and Engineering, Ohio State University, Columbus, Ohio 43201, USA. (e-mail:zhu.3440@osu.edu).}
\thanks{Casey Jameson and Zhexuan Gong are with the Department of Physics, Colorado School of Mines, Golden, Colorado 80401, USA. (e-mail:\{cwjameson,gong\}@mines.edu).}
\thanks{Alireza Goldar and Michael B. Wakin are with the Department of Electrical Engineering, Colorado School of Mines, Golden, Colorado 80401, USA. (e-mail:\{goldar,mwakin\}@mines.edu).}}

\begin{document}

\maketitle

\begin{abstract}
  While quantum state tomography (QST) remains the gold standard for benchmarking and verifying quantum devices, it requires an exponentially large number of measurements and classical computational resources for generic quantum many-body systems, making it impractical even for intermediate-size quantum devices. Fortunately, many physical quantum states often exhibit certain low-dimensional structures that enable the development of efficient QST. A notable example is the class of states represented by matrix product operators (MPOs) with a finite matrix/bond dimension, which include most physical states in one dimension and where the number of independent parameters describing the states only grows linearly with the number of qubits. Whether a sample efficient quantum state tomography protocol, where the number of required state copies scales only linearly as the number of parameters describing the state, exists for a generic MPO state still remains an important open question.

  In this paper, we answer this fundamental question affirmatively by using a class of informationally complete positive operator-valued measures (IC-POVMs)---including symmetric IC-POVMs (SIC-POVMs) and spherical $t$-designs--- focusing on sample complexity while not accounting for the implementation complexity of the measurement settings. For SIC-POVMs and spherical 2-designs, we show that the number of state copies to guarantee bounded recovery error of an MPO state with a constrained least-squares estimator depends on the probability distribution of the MPO under the POVM but scales only linearly with $n$ when the distribution is approximately uniform. For general spherical $t$-designs with a sufficiently large $t$, we prove that only a number of state copies proportional to the number of independent parameters in the MPO is sufficient for a guaranteed recovery of \emph{any} state represented by an MPO.
  In addition, we establish a sample complexity lower bound based on the Cram\'er–Rao bound, showing that, for strictly positive definite MPO states and unbiased estimators, a number of state copies proportional to the number of independent parameters in the MPO is also necessary, although the upper and lower bounds do not yet match in their dependence on the target accuracy. Furthermore, we propose a spectral method and a projected gradient descent (PGD) algorithm to solve the constrained least-squares problem, both equipped with rigorous recovery error guarantees.
\end{abstract}


\section{Introduction}
\label{intro}
The state of a quantum system composed of $n$ qudits (which are $d$-level quantum systems; qubits have $d=2$) is fully described by a density matrix $\vrho\in\C^{d^\nqbit \times d^\nqbit}$ with unit trace and is positive semidefinite (PSD)~\cite{nielsen2002quantum}. To reconstruct or estimate the density matrix, quantum measurements need to be performed on numerous identical copies of the state. Any such physical measurement is characterized by a Positive Operator-Valued Measure (POVM), which can be represented by a set of PSD matrices $\{\mA_1,\ldots,\mA_K\}$ that collectively sum to the identity matrix, i.e., $\sum_{k=1}^K \mA_k = \mId_{d^n}$. Each matrix $\mA_k$ ($k=1,\dots, K$) in the POVM corresponds to a potential measurement outcome, and the probability of obtaining that outcome is given by $p_k = \trace(\mA_k \vrho)$. The probabilistic nature of quantum measurements necessitates multiple measurements (say $M$) with the same POVM to obtain statistically accurate estimates $\wh p_k= \frac{f_k}{M}$ of each $p_k$, where $f_k$ denotes the number of times the $k$-th outcome is observed in the $M$ measurements.
Disregarding statistical errors, $\{p_k\}$ can be considered as $K$ linear measurements of the state~$\vrho$. From a machine learning perspective, we can refer to $\{p_k\}$ and their empirical estimates~$\{\wh p_k\}$ as population and empirical measurements of the state, respectively. The goal of quantum state tomography (QST) \cite{vogel1989determination} is to find the density matrix from empirical measurements with high accuracy. As the density matrix provides full information about the quantum state, QST remains the gold standard for benchmarking and verifying quantum devices.

Extensive research efforts, both in algorithmic development \cite{hradil1997quantum,vrehavcek2001iterative,blume2010optimal,granade2016practical,lukens2020practical,blume2012robust,faist2016practical,kyrillidis2018provable,brandao2020fast,
torlai2018neural,carleo2019machine,lohani2020machine} and theoretical analysis \cite{KuengACHA17,guctua2020fast,francca2021fast,liu2011universal, haah2017sample,voroninski2013quantum}, have been devoted to QST. For generic quantum states, the required state copies must grow exponentially with the number of qudits. This sample complexity can be improved for low-rank density matrices \cite{gross2010quantum,o2016efficient,haah2017sample,francca2021fast}, but it still grows exponentially with $n$. To implement QST effectively on current quantum computers with over 100 qubits~\cite{preskill2018quantum,arute2019quantum,chow2021ibm}, it is important to exploit more compact structures within practical states of interest.
A notable example is the class of states represented by matrix product operators (MPOs)\footnote{MPOs representing physical quantum states are also termed matrix product density operators (MPDOs), specifically designed for representing mixed states via density matrices. In contrast, matrix product states (MPSs) are used to represent pure states.  }. The density matrix elements of such states can be expressed as the following matrix product \cite{werner2016positive,noh2020efficient}
\begin{eqnarray}
\label{DefOfMPO}
\hspace{-0.5cm}\vrho(i_1 \cdots i_\nqbit, j_1 \cdots j_\nqbit)  =  \mX_1^{i_1,j_1} \mX_2^{i_2,j_2} \cdots \mX_\nqbit^{i_\nqbit,j_\nqbit},
\end{eqnarray}
where $i_k$ ($k=1,2,\cdots n$) labels a basis of the $k$-th qubit, and $\mX_\ell^{i_\ell,j_\ell}\in \C^{r_{\ell-1}\times r_\ell}$ with $r_0 = r_\nqbit = 1$. The dimensions $\vr = (r_1,\ldots,r_{n-1})$ are often called the {\it bond dimensions} \footnote{It is also common to simply call $\overline{r} = \max\{r_1,\ldots,r_{n-1}\}$ the bond dimension.} of the MPO in quantum physics, though we may also call them the {\it MPO ranks}. An MPO comprises $nd^2$ matrices, each with a dimension of at most $\ol r \times \ol r$ where the matrix dimension $\ol r = \max_\ell r_\ell$ determines the efficiency of the representation. Thus, an MPO can be efficiently represented with $O(nd^2\ol r^2)$ number of parameters. The decomposition \eqref{DefOfMPO} is also mathematically equivalent to the tensor train (TT) decomposition used for compact representation of large tensors \cite{Oseledets11}.

Examples of states that can be represented by MPOs include ground states of most quantum systems in one dimension \cite{eisert2010}, general Hamiltonians with decaying long-range interactions \cite{pirvu2010matrix}, and those generated by noisy quantum computers, where the noise could also limit the amount of quantum entanglement and thus enable an efficient state representation \cite{noh2020efficient}.
Several approaches, such as maximum likelihood estimation \cite{baumgratz2013scalable,choi2022single} and tensor train cross approximation \cite{lidiak2022quantum}, have empirically demonstrated the scalable number of state copies when employing MPO in QST. Nevertheless, the theoretical analysis of this advantage has remained ambiguous in previous studies. Recently, the work \cite{qin2024quantum} successfully established a theoretical guarantee on the recovery error of an MPO state using multiple Haar random projective measurements. To be specific, for any given $\epsilon>0$, \cite{qin2024quantum} proves that with  $O(nd^2 \ol{r}^2\log n)$ different sets of Haar random projective measurement bases and $O(n^3 d^2 \ol{r}^2\log n/\epsilon^2)$ total state copies, a properly constrained least-squares minimization using the empirical measurements stably can recover the ground-truth state with $\epsilon$-closeness in the Frobenius norm with high probability. While this result formally shows that MPOs can be reconstructed using polynomial number of state copies with regard to $n$, it is not sample efficient as the scaling $n^3$ is much larger than the degrees of freedom ($O(n)$) in the MPOs. In addition, this work \cite{qin2024quantum} provides a projected gradient descent (i.e., iterative hard thresholding) algorithm for performing the QST in practice, but does not provide a formal guarantee for that algorithm.

\subsection{Main contributions}
\label{sec: main contribution}

Our work aims to address two fundamental questions regarding the QST of states represented by MPOs:
\begin{enumerate}
\item \it Is there a measurement setting that can provably recover an MPO with a number of state copies proportional to the number of parameters in the MPO?
\item \it If so, can we develop a provably correct algorithm to reconstruct MPOs from the obtained measurements?
\end{enumerate}

\begin{table*}[!ht]
\renewcommand{\arraystretch}{1.8}
\begin{center}
\caption{Comparison with existing work for estimating $n$-qudits MPOs (with bond dimension $\ol r$) to achieve $\epsilon$-closeness in the Frobenius norm. The result for SIC/2-designs depends on $\gamma_{\max}(\vrho)=K\cdot \max_k p_k$, i.e., the probability distribution of measuring the state $\vrho$ with the corresponding POVM.}
\label{Comparison among minimum M for various POVMs}

\vspace{.1in}
\begin{tabular}{|l|l|l|}
\hline
\multicolumn{1}{|c|}{POVM}                                                                                                   & \multicolumn{1}{c|}{$\#$ POVMs}              & \multicolumn{1}{c|}{$\#$ total state copies}                                                            \\ \hline
Haar random (\cite[Theorem 5]{qin2024quantum})                                        & $\Omega(nd^2 \ol{r}^2 \log n)$ & $O(\frac{n^3 d^2 \ol{r}^2\log n }{\epsilon^2})$                  \\ \hline
SIC/$2$-designs (\Cref{Recovery error of various quantum POVM})                            &  \hspace{1cm} $1$                            & $O(\frac{n d^2 \gamma_{\max}(\vrho)\ol{r}^2 \log n}{\epsilon^2})$ \\ \hline
$t$-designs satisfying \eqref{requirement of t designs in the main paper 1} (\Cref{Recovery error of various t-designs POVM 3}) & \hspace{1cm} $1$                            & $O(\frac{n d^2 \ol{r}^2 \log n}{\epsilon^2})$                    \\ \hline
\end{tabular}
\end{center}
\end{table*}

We provide positive answers to both the above questions by utilizing informationally complete POVMs (IC-POVM). An IC-POVM is defined by the property that every quantum state is uniquely determined by its measurement statistics. Beyond QST, IC-POVMs with special ties find applications in quantum cryptography \cite{renes2004equiangular}, quantum fingerprinting \cite{scott2005optimal}, and are relevant to foundational studies of quantum mechanics \cite{caves2002unknown, konig2005finetti}. It eliminates the inherent randomness in Haar random POVMs that are often used to establish efficient sample complexity \cite{o2016efficient,haah2017sample}. Without this inherent randomness, we can more easily analyze the effects of statistical errors in the recovery error bound and the design provable optimization algorithms.

We begin by studying recovery guarantees for MPOs from empirical quantum measurements (physical measurements containing statistical errors) obtained using a class of IC-PVOMs including the rank-one symmetric IC-POVM (SIC-POVM) \cite{renes2004symmetric} and more broadly spherical $t$-design POVMs that are induced by spherical $t$-designs \cite{zauner1999quantum,dall2014accessible}. {\it The first contribution of this paper---presented in \Cref{sec:stable recovery}---is a new sample complexity bound (refers to the number of state copies) that matches the number of parameters in MPOs to guarantee bounded recovery error for a particular estimator---the solution to a constrained least-squares optimization problem.} Specifically, for SIC-POVM and $2$-design POVMs,  our result shows that $M = O(n d^2 \gamma_{\max}(\vrho)\ol{r}^2 \log n/\epsilon^2)$ number of state copies is sufficient to achieve $\epsilon$-accuracy in recovering the ground-truth state $\vrho$ in the Frobenius norm, where $\gamma_{\max}(\vrho) = K \cdot \max_{k=1}^K p_k$ captures the uniformity of the probabilistic distribution of the measurement. The term $\gamma_{\max}(\vrho)$ can be as small as 1, but in some cases, it may be as large as $d^n$; however, for many cases it is only polynomially large in $n$. The dependence of $\gamma_{\max}(\vrho)$ could be eliminated by using spherical $t$-designs with a sufficiently large $t$, for which we only need  $M = O(n d^2 \ol{r}^2 \log n/\epsilon^2)$ to ensure the sufficient condition.  Excluding the logarithmic term, this bound exhibits linear dependence on $n$, $d^2$, and $\ol r^2$, aligning precisely with the degrees of freedom $O(nd^2\ol r^2)$ in MPOs and is thus sample efficient. This result improves upon previous work \cite{qin2024quantum} mentioned above using Haar random projective measurement bases that requires $\Omega(nd^2 \ol{r}^2\log n)$ POVMs and $O(n^3 d^2 \ol{r}^2\log n/\epsilon^2)$ total number of state copies. In \cite{qin2024quantum}, the measurement operator itself is random, and the two sources of randomness---the Haar random measurement bases and the statistical errors---are controlled separately. In particular, controlling the statistical error requires a uniform bound on the random measurement coefficients over all outcomes of all Haar random bases, which is obtained by a union bound over exponentially many outcomes and introduces the additional $n^2$ factor. In contrast, the IC-POVMs considered in this paper are deterministic and exactly preserve the distance between any two states, so the statistical error is the only source of randomness. This allows us to apply a concentration inequality for multinomial random variables directly, without any union bound over the measurement outcomes, so that only the covering argument over the set of MPOs contributes to the sample complexity, which therefore scales linearly with $n$. We summarize the main results and comparison with our previous result \cite{qin2024quantum} in Table~\ref{Comparison among minimum M for various POVMs}. Moreover, physical MPO states are often approximately low-rank, allowing the previous analysis to be extended to the trace distance. Specifically, to ensure a bounded recovery error in terms of trace distance for rank-$s$ MPO states, the required number of state copies remains proportional to the degrees of freedom, $O(nsd^2\ol r^2)$, when using  spherical $t$-designs with $t \geq 2$. We further establish a Cram\'er–Rao-type lower bound showing that, for strictly positive definite MPO states and unbiased estimators, achieving an expected trace-norm accuracy $\zeta$ requires $\Omega(nd^2\ol r^2/\zeta)$ state copies. This indicates that the linear dependence on the degrees of freedom $O(nd^2\ol r^2)$ of an MPO in our upper bound is necessary. Nevertheless, the upper and lower bounds are established under different assumptions and exhibit different dependence on the target accuracy ($1/\epsilon^2$ versus $1/\zeta$), and thus they do not yet provide a fully matched characterization of the optimal sample complexity.

{\it The second contribution of this paper---presented in \Cref{sec: algorithm}---is the development of a spectral method and a projected gradient descent (PGD) algorithm for solving the aforementioned constrained least-squares optimization problem, both equipped with rigorous recovery error guarantees.} The spectral method is a one-step estimator, with a projection step incorporated, that achieves a sub-optimal error bound of order $O(n^2)$, whereas PGD is an iterative optimization method that performs a gradient update followed by a projection onto the set of MPOs with unit trace, reducing the error bound to order $O(n)$ via a carefully designed projection step.  Even without the constraint on the trace, computing the optimal projection onto the set of MPOs is challenging and could be NP-hard \cite{hillar2013most}. We employ a sequential approach called TT-SVD \cite{Oseledets11}---developed for computing TT decomposition which is equivalent to MPO---for computing an approximate projection. While TT-SVD is not optimal, it is computationally efficient and comes with a sub-optimal guarantee. Numerical experiments in Section~\ref{sec: experimental results} validate the effectiveness of both the spectral method and PGD.

In this paper, we primarily focus on global spherical $t$-design POVMs with $t\ge 2$, as we theoretically demonstrate that stable recovery of constrained least-squares optimization for all MPO states can be achieved using a polynomial number of state copies. In practice, however, local measurements are widely employed in quantum state tomography due to their efficient implementation and lower computational complexity. Notably, when the Hilbert-Schmidt distance for local marginals is bounded, \cite{fanizza2023learning} has shown that the recovery error in terms of trace distance grows polynomially with $n$ for a translation-invariant state on an infinite one-dimensional chain. Nevertheless, the recent work in \cite{casey2024optimal} reveals that recovering long-range entangled states--such as the general Greenberger-Horne-Zeilinger (GHZ) state, which admits an MPO representation--requires an exponential number of state copies when relying on local measurements. This finding highlights that, both theoretically and practically, local measurements are insufficient to guarantee the recovery of all MPO states using a polynomial number of state copies.

We note that our work is centered on the fundamental question of whether the number of state copies required for QST can be efficient (i.e. proportional to the intrinsic degrees of freedom) for recovering a structured quantum many-body state. We investigate this problem using IC-POVMs, focusing on sample complexity while not accounting for the implementation complexity of the measurement settings. In practice, the SIC-POVM or the spherical $t$-designs with $t\ge 2$ have no known efficient implementation using local quantum circuits. To partially address this practicality concern, we provide numerical evidence that local SIC-POVMs, which can be implemented efficiently and readily on current quantum hardware, can be used to perform QST for typical MPO states with $O(n)$ state copies using the proposed PGD algorithm in Section \ref{sec: experimental results}.

\subsection{Notation}
\label{sec:notation}

We use calligraphic letters (e.g., $\mathcal{X}$) to denote tensors,  bold capital letters (e.g., $\bm{X}$) to denote matrices,  bold lowercase letters (e.g., $\bm{x}$) to denote column vectors, and italic letters (e.g., $x$) to denote scalar quantities. Elements of matrices and tensors are denoted in parentheses. For example, $\calX(i_1, i_2, i_3)$ denotes the element in position
$(i_1, i_2, i_3)$ of the order-3 tensor $\calX$. The calligraphic letter $\calA$ is reserved for the linear measurement map.  For a positive integer $K$, $[K]$ denotes the set $\{1,\dots, K \}$. The superscripts $(\cdot)^\top$ and $(\cdot)^\dagger$ denote the transpose and Hermitian transpose, respectively \footnote{As is conventional in the quantum physics literature (but not in information theory and signal processing), we use $(\cdot)^\dagger$ to denote the Hermitian transpose.}. For two matrices $\mA,\mB$ of the same size, $\innerprod{\mA}{\mB} = \trace(\mA^\dagger\mB)$ denotes the inner product between them.
$\|\mX\|$ and $\|\mX\|_F$ respectively represent the spectral norm and Frobenius norm of $\mX$.
$\sigma_{i}(\mX)$ is the $i$-th singular value of $\mX$.
For a vector $\va$ of size $N\times 1$, its $l_n$-norm is defined as $||\va||_n=(\sum_{m=1}^{N}|a_m|^n)^\frac{1}{n}$.
For two positive quantities $a,b\in \real$, the inequality $b\lesssim a$ or $b = O(a)$ means $b\leq c a$ for some universal constant $c$; likewise, $b\gtrsim a$ or $b = \Omega(a)$ represents $b\ge ca$ for some universal constant $c$.
To simplify notations in the following sections, for an order-$n$ TT format with ranks $(r_1,\dots, r_{n})$, we define $\ol r=\max_{i=1}^{n-1} r_i$.

\section{Informationally Complete POVM}
\label{sec: Stble embedding of density operator}

As described in the introduction, the population measurements $\{p_k\}$ of the density matrix $\vrho$ from a POVM $\{\mA_k\}$ are linear measurements that can be described through a linear map $\calA:  \C^{d^n\times d^n} \rightarrow \R^{K}$ of the form
\e
\calA(\vrho)= \begin{bmatrix}
          \< \mA_1, \vrho  \> \\
          \vdots \\
          \< \mA_{K}, \vrho  \>
        \end{bmatrix}.
\label{eq:linear map A}\ee
If a POVM consists of at least $d^{2n}$ matrices such that we can form exactly $d^{2n}$ linearly independent matrices by linear combination,  it is said to be an informationally complete POVM (IC-POVM) \cite{prugovevcki1977information,busch1989determination,peres1997quantum}. IC-POVMs consisting of exactly $d^{2n}$ elements are called minimal \cite{weigert2006simple}. In this section, we introduce some common IC-POVMs, including symmetric IC-POVM and spherical $t$-designs, and present properties of the corresponding linear mappings $\calA$.

\paragraph*{Symmetric Informationally Complete POVM (SIC-POVM)} \ \ A particularly attractive and potentially useful IC-POVM is one that is {\it symmetric}, i.e., it has equal pairwise inner products between any pair of the POVM elements~\cite{renes2004symmetric,flammia2006sic}. A POVM that is informationally complete, minimal, and symmetric is called SIC-POVM.

\begin{defi} (\cite{geng2020minimal})
\label{defi of SIC POVM}
A Symmetric Informationally Complete POVM (SIC-POVM) consists of $d^{2n}$ rank-one POVM elements $\{\mA_k\}$ that satisfies $(i)$ each with an equivalent trace $\trace(\mA_k) = \frac{1}{d^n}$ for all $k\in [d^{2n}]$, and $(ii)$ every pair of measurement operators satisfies $\<\mA_k, \mA_j \> = \begin{cases}
    \frac{1}{d^{2n}}, & k = j,\\
    \frac{1}{d^{2n}(d^n + 1)}, & k \neq j.
  \end{cases}$
\end{defi}

The set of vectors comprising an SIC-POVM has been extensively studied in various contexts under different names \cite{renes2004symmetric}, such as equiangular lines in linear algebra \cite{greaves2016equiangular} and equalangular tight frames in information theory \cite{sustik2007existence}. In quantum information, it has also found extensive applications, including quantum state tomography \cite{scott2006tight} and quantum cryptography \cite{durt2008wigner}.  SIC-POVMs have been numerically constructed for matrix dimensions ranging from $2$ to $151$~\cite{scott2010symmetric, samuel2024sic}, and for many higher dimensions, up to $39,604$ \cite{bengtsson2024sic}. While it is widely believed that SIC-POVMs exist in all dimensions, no explicit construction of the SIC-POVM is known for an arbitrary dimension of the Hilbert space. The following result establishes stable embeddings of any Hermitian matrices from SIC-POVM in terms of $\|\calA(\vrho)\|_2^2$.
\begin{lemma}
\label{l2 norm of SIC-POVM RIP}
Suppose that $\{\mA_k \}_{k=1}^{d^{2n}}$ form an SIC-POVM. Then for any  Hermitian matrix $\vrho\in\C^{d^n\times d^n}$, $\calA(\vrho)$ satisfies
\begin{eqnarray}
\label{The l2 norm of A(rho)}
\hspace{-0.6cm}\| \calA(\vrho)\|_2^2 = \sum_{k=1}^{d^{2n}}\< \mA_k, \vrho  \>^2 = \frac{\|\vrho\|_F^2 + (\trace(\vrho))^2}{d^n(d^n + 1)}.
\end{eqnarray}
\end{lemma}

The proof is provided  in Appendix~\ref{Proof of RIP SIC POVM in Appe}. By using tools from compressive sensing literature~\cite{donoho2006compressed, candes2006robust, candes2008introduction, recht2010guaranteed, eftekhari2015new}, the previous work~\cite{qin2024quantum} provides a similar stable embedding result from Haar random projective measurement bases, but only for MPOs. In contrast, with POVM being informationally complete, \eqref{The l2 norm of A(rho)} holds for any Hermitian matrices without relying on any randomness in the measurement operator. This allows us to develop a sample-efficient sample complexity bound in \Cref{sec:stable recovery} and iterative algorithms with guaranteed convergence in \Cref{sec: algorithm}. As an immediate consequence of \Cref{l2 norm of SIC-POVM RIP}, for any two density matrices $\vrho_1$ and $\vrho_2$, the difference $\vrho_1-\vrho_2$ remains Hermitian, and hence
\begin{eqnarray}
\label{The l2 norm of diff SIC POVM}
\| \calA(\vrho_1-\vrho_2)\|_2^2 & =& \frac{\|\vrho_1-\vrho_2\|_F^2 + (\trace(\vrho_1-\vrho_2))^2}{d^n(d^n + 1)}\nonumber\\
&=& \frac{\|\vrho_1-\vrho_2\|_F^2 }{d^n(d^n + 1)},
\end{eqnarray}
where we use the property that $\trace(\vrho_1) = \trace(\vrho_2) = 1$. Thus, SIC-POVM not only ensures distance measurements ($\calA(\vrho_1) \neq \calA(\vrho_2) $) as long as $\vrho_1\neq \vrho_2$, but also exactly preserves the distance of the two states in the measurement space, up to a universal scaling $1/d^n(d^n + 1)$.

\paragraph*{$t$-design POVMs} \ \  A more general set of IC-POVM are those induced by spherical $t$-designs.
\begin{defi}
\label{definition_of_T_Design} (Spherical $t$-designs \cite{matthews2009distinguishability,KuengACHA17}). A finite set $\{\vw_k  \}_{k=1}^K\subset \C^{d^n}$ of normalized vectors is called a  spherical quantum $t$-design if \footnote{ $K\geq C_{d^n + \lfloor t/2 \rfloor - 1}^{\lfloor t/2 \rfloor} \cdot C_{d^n + \lceil t/2 \rceil - 1}^{\lceil t/2 \rceil} $ is  necessary to form a spherical $t$-design\cite{scott2006tight,gross2015partial}.}
\begin{eqnarray} \label{the definition of t designs}
     \frac{1}{K}\sum_{k=1}^K (\vw_k\vw_k^\dagger)^{\otimes s}  = \int (\vw\vw^\dagger)^{\otimes s} d\vw
\end{eqnarray}
holds for any $s\leq t$, where the integral on the right hand side is taken with respect to the Haar measure on the
complex unit sphere in $\C^{d^n}$.
\end{defi}
When $s=1$, we have $\frac{1}{K}\sum_{k=1}^K \vw_k\vw_k^\dagger = \int \vw\vw^\dagger d\vw = \frac{1}{d^n}\mId$, and hence $\mA_k = \frac{d^n}{K} \vw_k\vw_k^\dagger, k = 1,\dots, K$ form a rank-one POVM. For simplicity, we call such an induced POVM $\{\mA_k = \frac{d^n}{K} \vw_k\vw_k^\dagger\}$ as a $t$-design POVM. SIC-POVM is a special instance of $2$-design POVMs. In \eqref{the definition of t designs}, we adopt the setting of uniform weights ($1/K$ for each $\vw_k$, which is mostly commonly used in practice) to simplify the analysis, but one can consider a more general scenario of varied weights \cite{dall2014accessible}. Unlike SIC-POVM, $t$-designs always exist and can, in principle, be constructed in any dimension and for any $t$ \cite{seymour1984averaging, bajnok1992construction},  although in some cases these constructions can be inefficient, as they require vector sets of exponential size \cite{hayashi2005reexamination}.

The following result generalizes \Cref{l2 norm of SIC-POVM RIP} to $t$-design POVMs.
\begin{lemma}
\label{l2 norm of approximate 2-designs RIP}
Suppose that $\{\vw_k  \}_{k=1}^K\subset \C^{d^n}$ forms a $t$-design ($t\ge 2$). Let $\calA$ be the linear map in \eqref{eq:linear map A} corresponding to the induced POVM $\{\mA_k = \frac{d^n}{K} \vw_k\vw_k^\dagger\}$.
Then for arbitrary Hermitian matrix $\vrho\in\C^{d^n\times d^n}$, $\calA(\vrho)$ satisfies
\begin{eqnarray}
\label{The l2 norm of A(rho) approximate 2_designs}
\|\calA(\vrho) \|_2^2 =  \frac{d^n(\|\vrho\|_F^2 + (\trace({\vrho}))^2)}{K(d^n + 1)}.
\end{eqnarray}
\end{lemma}
The proof is given in Appendix~\ref{Proof of RIP 2 designs in Appe}. Similar to \eqref{The l2 norm of diff SIC POVM}, it follows from \Cref{The l2 norm of A(rho) approximate 2_designs} that for any two density matrices $\vrho_1$ and $\vrho_2$, $\calA(\vrho_1-\vrho_2)$ satisfies
\begin{eqnarray}
\label{The l2 norm of diff approximate 2 designs POVM}
&\!\!\!\!\!\!\!\!&\| \calA(\vrho_1-\vrho_2)\|_2^2 \nonumber\\
&\!\!\!\!=\!\!\!\!&  \frac{d^n(\|\vrho_1-\vrho_2\|_F^2 + (\trace({\vrho_1-\vrho_2}))^2)}{K(d^n + 1)} \nonumber\\
&\!\!\!\!=\!\!\!\!&  \frac{d^n\|\vrho_1-\vrho_2\|_F^2 }{K(d^n + 1)}.
\end{eqnarray}

\section{Stable Recovery of Matrix Product Operators}
\label{sec:stable recovery}

Building upon the results that a  $t$-design POVM $\{ \mA_1,\dots, \mA_K \}\in\C^{d^n\times d^n}$ provides unique population measurements for any state, in this section, we investigate the stable recovery for MPOs of form \eqref{DefOfMPO}  from empirical measurements acquired through such POVMs. Specifically, for an MPO $\vrho^\star\in\C^{d^n \times d^n}$ with MPO ranks $(r_1,\dots, r_{n-1})$, we use the POVM to measure it $M$ times and take the average of the outcomes to generate empirical probabilities
\begin{eqnarray}
\wh p_{k} = \frac{f_{k}}{M}, \ k=1,\dots,K,
\label{eq:empirical-prob in SIC-POVM}
\end{eqnarray}
where $f_{k}$ denotes the number of times the $k$-th output is observed. We denote ${\wh\vp}= \begin{bmatrix}
          \wh{p}_{1} & \cdots &
          \wh{p}_{K}
        \end{bmatrix}^\top$
as the empirical measurements obtained by the POVM, which are unbiased estimators of the population measurements $\vp = \begin{bmatrix} p_1 & \cdots & p_K\end{bmatrix}^\top= \calA(\vrho^\star)$, where $\calA$ is the induced linear measurement operator formally defined in \eqref{eq:linear map A}. To simplify the notation for the analysis, we  denote by $\veta = \begin{bmatrix}
          \eta_{1},
          \cdots,
          \eta_{K}\end{bmatrix}^\top = \wh{\vp} -  \vp$ the statistical error between the empirical probabilities and population probabilities.

With empirical measurements $\wh \vp$, we consider minimizing the following constrained least squares objective:
\begin{eqnarray}
    \label{The loss function in QST for SIC-POVM}
    \wh{\vrho} = \argmin_{\vrho\in\setX_{\vr}}\|\calA(\vrho) - {\wh\vp}\|_2^2,
\end{eqnarray}
where we denote by $\setX_{\vr}$ the set of MPOs with MPO  ranks $\vr = (r_1,\dots, r_{n-1})$:
\begin{eqnarray}
\label{SetOfMPO}
\begin{split}
&\setX_{\vr}= \Big\{ \vrho\in\C^{d^n\times d^n}:\ \vrho = \vrho^\dagger, \trace(\vrho) = 1, \\ &\vrho(i_1 \cdots i_\nqbit, j_1\cdots j_\nqbit) = \mX_1^{i_1,j_1} \mX_2^{i_2,j_2} \cdots \mX_\nqbit^{i_\nqbit,j_\nqbit}, \\
& \mX_\ell^{i_\ell,j_\ell}\in\C^{r_{\ell-1}\times r_\ell}, \ell=1,\dots,n-1, r_0=r_n=1\Big\}.
\end{split}
\end{eqnarray}

It is worth noting that one could consider imposing additional structure on the factors ${\mX_{\ell}^{i_\ell,j_\ell}}$ to ensure that $\vrho$ is PSD, as done in~\cite[eq. (3)]{verstraete2004matrix}. However, the condition in \cite[eq. (3)]{verstraete2004matrix} is only sufficient, not necessary, for ensuring $\vrho$ is PSD and adding PSD constraint does not significantly reduce the number of degrees of freedom of elements in the set $\setX_{\vr}$. Hence, we focus on the set \eqref{SetOfMPO}, which encompasses not only PSD matrices but also non-PSD matrices.

Assuming a global solution to \eqref{The loss function in QST for SIC-POVM} can be found, our objective is to study the performance of $\wh\vrho$, and examine how the recovery error $\|\wh{\vrho}  -  \vrho^\star\|_F$ scales with the size of the MPO (specifically in relation to the number of qudits $n$) and the number of measurements $M$. In order to achieve a stable estimation of the state with a polynomial number of measurements relative to $n$, we expect the recovery error exhibits polynomial growth rather than exponential growth with respect to $n$. Next section will provide iterative algorithm to solve \eqref{The loss function in QST for SIC-POVM}.

\paragraph*{Main results} \ \ We now provide a formal analysis of the recovery error $\|\wh{\vrho}  -  \vrho^\star\|_F$. Since both SIC-POVM  can be viewed as a specific instance of   $t$-design POVMs, the analysis will primarily focus on $t$-design POVM. Using the fact that $\wh{\vrho}$ is a global solution of \eqref{The loss function in QST for SIC-POVM} and the fact that $\vrho^\star \in \setX_{\bar{r}}$, we have
\begin{eqnarray}
    \label{whrho and rho star relationship}
    0 & \!\!\!\!\leq \!\!\!\!& \|\calA(\vrho^\star) - \wh{\vp} \|_2^2  - \|\calA(\wh{\vrho}) - \wh{\vp} \|_2^2\nonumber\\
    &\!\!\!\! = \!\!\!\!&\|\calA(\vrho^\star)-\calA(\vrho^\star) - \veta\|_2^2 - \|\calA(\wh{\vrho})-\calA(\vrho^\star) - \veta\|_2^2\nonumber\\
    &\!\!\!\! = \!\!\!\!& 2\<\calA(\vrho^\star)+\veta, \calA(\wh{\vrho} - \vrho^\star)  \> + \|\calA(\vrho^\star)\|_2^2 - \|\calA(\wh{\vrho})\|_2^2\nonumber\\
    &\!\!\!\! = \!\!\!\!& 2\<  \veta, \calA(\wh{\vrho} - \vrho^\star) \> - \|\calA(\wh{\vrho} - \vrho^\star)\|_2^2,
\end{eqnarray}
which further implies that
\begin{eqnarray}
    \label{whrho and rho^star relationship_1}
    \|\calA(\wh{\vrho} - \vrho^\star)\|_2^2 \leq 2\<  \veta, \calA(\wh{\vrho} - \vrho^\star) \>.
\end{eqnarray}

According to \eqref{The l2 norm of diff approximate 2 designs POVM} for $t$-design POVM, we have
\begin{eqnarray}
    \label{summary of lower bound of difference}
    \|\calA(\wh{\vrho} - \vrho^\star)\|_2^2 = \frac{d^n\|\wh{\vrho} - \vrho^\star\|_F^2 }{K(d^n + 1)}.
\end{eqnarray}

The challenge part is to deal with the right-hand side of \eqref{whrho and rho^star relationship_1}. A straightforward application of the Cauchy-Schwarz inequality $\<  \veta, \calA(\wh{\vrho} - \vrho^\star) \> \le \|\veta\|_2 \cdot \|\calA(\wh{\vrho} - \vrho^\star)\|_2$ is insufficient for providing a tight result as $\|\veta\|_2$ scales as $\frac{1}{\sqrt{M}}$ according to \cite[eq. (35)]{qin2024quantum}. Instead, we leverage the randomness of $\veta$ and employ the following concentration bound for multinomial random variables.
\begin{lemma} (\Cref{Refined multinomial concentration psi1} in Appendix~\ref{sec:app-aux})
\label{General bound of multinomial distribution Q cases1}
Suppose $(f_{1},\dots, f_{K})$ are mutually independent and follow the multinomial distribution $\operatorname{Multinomial}(M,\vp)$  where $\sum_{k=1}^{K}f_{i} =M $ and $\vp = [p_{1},\cdots, p_{K}]$. Let $a_{1},\dots, a_{K}$ be fixed. Then, for any $\tau>0$ satisfying
\begin{eqnarray}
\label{requirement in the concentration inequality main paper}
\frac{\sum_{k=1}^K a_k^2p_k \exp \left(\frac{(a_k)_+\tau}{2\sum_{j=1}^K a_j^2p_j}\right)} {\sum_{k=1}^K a_k^2p_k} \leq 2,
\end{eqnarray}
where $(a_k)_+ = \max\{ a_k, 0 \}$, we have
\begin{eqnarray}
\label{General bound of multinomial distribution for all constant Q cases1}
\P{\sum_{k=1}^K a_k\Big(\frac{f_k}{M}-p_k\Big)>\tau} \leq \exp \left( -\frac{M\tau^2}{4\sum_{k=1}^K a_k^2p_k} \right).
\end{eqnarray}
\end{lemma}

We proceed by considering the following two cases, $2$-designs and general $t$-designs.

\paragraph*{Case $(i)$: $2$-designs.} \ \
There are three potential technical difficulties in directly applying \Cref{General bound of multinomial distribution for all constant Q cases1} to bound the term $\<  \veta, \calA(\wh{\vrho} - \vrho^\star) \>$ by plugging in $a_k = \<\mA_k, \wh\vrho - \vrho^\star \>$: (1) $a_k$ in not independent to the random variables $\{f_k\}$ since $\wh \vrho$ depends on the empirical measurements, (2) we need to upper bound a term of form $\sum_{k = 1}^{K}\<\mA_k, \vrho - \vrho^\star \>^2p_{k}$ originating from the exponential term in \eqref{General bound of multinomial distribution for all constant Q cases1}, and (3) the condition in \eqref{requirement in the concentration inequality main paper} needs to be verified. We will address the first challenge by using a covering argument. For the second challenge, to utilize the upper bound for $\|\calA(\vrho - \vrho^\star)\|_2^2 = \sum_{k = 1}^{K}\<\mA_k, \vrho - \vrho^\star \>^2$ in Section~\ref{sec: Stble embedding of density operator}, we capture the largest probabilities by
\begin{eqnarray}
    \label{requirement in _global SIC-POVM main paper}
    \gamma_{\max}(\vrho^\star) := K \cdot \max_k p_k =  K \cdot \max_{k} \langle \mA_k, \vrho^\star \rangle.
\end{eqnarray}
Since $\max_k p_k \ge \frac{1}{K}$, $\gamma_{\max}(\vrho^\star) \ge 1$ always holds. When the probability distribution is uniform, $\gamma_{\max}(\vrho^\star)$ can be small as $O(1)$; particularly, $\gamma_{\max}(\vrho^\star) = 1$ when $p_1 = \cdots = p_K = 1/K$. On the other hand, if the probability distribution is spiky, $\gamma_{\max}(\vrho^\star)$ could approach $K$.
For the third challenge, verifying the condition in \eqref{requirement in _global SIC-POVM main paper} also requires control of the smallest measurement probabilities. To this end, we define
\begin{eqnarray}
    \label{requirement in _global SIC-POVM smallest main paper}
    \gamma_{\min}(\vrho^\star) := K \cdot \min_k p_k =  K \cdot \min_{k} \langle \mA_k, \vrho^\star \rangle.
\end{eqnarray}
Since $\min_k p_k \leq 1/K$, we have $\gamma_{\min}(\vrho^\star)\leq 1$. In particular, $\gamma_{\min}(\vrho^\star)=1$ for the uniform distribution, whereas $\gamma_{\min}(\vrho^\star)$ can approach zero when the probability distribution is highly concentrated. Moreover, since the covering argument requires the condition in \eqref{requirement in the concentration inequality main paper} to hold uniformly over all normalized MPO directions, we also need to control the largest measurement coefficient over such directions, which is captured by
\begin{eqnarray}
\label{requirement of t designs in the main paper 12}
\beta := K \max_{k\in[K]}\ \max_{\vrho\in\ol\setX_{4\vr+1}}\big|\<\mA_k,\vrho\>\big|,
\end{eqnarray}
where $\ol\setX_{4\vr+1} = \Big\{ \vrho\in\C^{d^n\times d^n}:\ \vrho = \vrho^\dagger,  \|\vrho\|_F =  1,  \trace(\vrho) = 0,  \vrho \ \text{is an MPO with MPO  ranks} \ (4r_1+1,\dots, 4r_{n-1}+1) \Big\}$. Here, the MPO ranks $4\vr+1$ arise from the covering argument, where the difference of two MPOs with MPO ranks $\vr$ has MPO ranks at most $2\vr$, taking its Hermitian part doubles the MPO ranks, and removing the trace component adds one.

We now summarize the main result as follows.
\begin{theorem}[Stable recovery with $2$-designs]
\label{Recovery error of various quantum POVM}
Suppose $\{\mA_1, \ldots, \mA_{K}\}$ form a set of    $2$-design POVMs.  Consider an MPO state $\vrho^\star\in\C^{d^n \times d^n}$ with MPO ranks $(r_1,\dots, r_{n-1})$. Then, with the empirical measurements $\wh \vp$ obtained by measuring the state $M$ times using the POVM and probability $1 - e^{- \Omega( n d^2 \ol{r}^2\log n)}$, the solution $\wh \vrho$ of the constrained least squares \eqref{The loss function in QST for SIC-POVM} satisfies
\begin{eqnarray}
\label{final conclusion of recovery error1}
\|\wh{\vrho} - \vrho^\star\|_F\lesssim \sqrt{\frac{n d^2 \ol{r}^2\gamma_{\max}(\vrho^\star) \log n}{M}},
\end{eqnarray}
provided that $M = \Omega\big(  nd^2 \ol{r}^2 \beta^2\gamma_{\max}(\vrho^\star)\log n/\gamma_{\min}^2(\vrho^\star) \big)$.
\end{theorem}
The proof is given in Appendix~\ref{Proof of recovery error of SIC-POVM in Appe}. \Cref{Recovery error of various quantum POVM} guarantees a stable recovery of the ground-truth state with a $2$-design POVM when the number of state copies $M = O(n d^2 \gamma_{\max}(\vrho^\star) \ol{r}^2 \log n/\epsilon^2)$, which achieves an $\epsilon$-accurate recovery. This condition highlights that the growth of $M$ is only polynomial in terms of the number of qudits $n$ when $\gamma_{\max}(\vrho^\star)$ is small. In particular, when $\gamma_{\max}(\vrho^\star) = O(1)$, we can achieve an efficient theoretical bound for $M$ that matches $O(nd^2\ol{r}^2)$, the number of degrees of freedom of MPOs. On the other hand, when $\gamma_{\max}(\vrho^\star)$ is polynomial in terms of $n$, the order of $M$ still remains polynomial. In addition, we emphasize that the condition $M = \Omega\big(  nd^2 \ol{r}^2 \beta^2\gamma_{\max}(\vrho^\star)\log n/\gamma_{\min}^2(\vrho^\star) \big)$ is primarily introduced to verify \eqref{requirement in the concentration inequality main paper} uniformly over all normalized MPO directions and is not necessarily intrinsic to the recovery guarantee itself. Here, $\gamma_{\max}(\vrho^\star) / \gamma_{\min}^2(\vrho^\star)$ characterizes the degree of non-uniformity of the measurement probability distribution, which equals one for the uniform distribution and increases as the distribution becomes more uneven, while $\beta$ characterizes the coherence between the POVM and the set of MPOs, whose range is discussed after \Cref{Recovery error of various t-designs POVM 3}.

As mentioned above, $\gamma_{\max}(\vrho)\ge 1$ always holds. To further illustrate possible values for $\gamma_{\max}(\vrho)$ in the sample complexity $M = O(n d^2 \gamma_{\max}(\vrho^\star) \ol{r}^2 \log n/\epsilon^2)$, we consider the case of SIC-POVM. First consider an identity density matrix $\vrho = \frac{1}{d^n}\mId_{d^n} = \frac{1}{d}\mId_d\otimes \cdots \otimes \frac{1}{d}\mId_d$ which can be written as an MPO with bond dimension $1$. For this case,  $\<\mA_k, \vrho \> = \frac{1}{d^n}\trace(\mA_k) = \frac{1}{d^{2n}}$ for all $k$, and hence $\gamma_{\max}(\vrho) = 1$. On the other hand, when $\vrho = d^n \mA_k$ for some $k$ (note that such a rank-one density matrix may not be an MPO), we have $\<\mA_k, \vrho \> = \frac{1}{d^n}$ and $\<\mA_j, \vrho \> = \frac{1}{d^n(d^n+1)}$ for all $j\neq k$, and hence $\gamma_{\max}(\vrho) = d^n$. Indeed, we can prove that $\gamma_{\max}(\vrho) \le d^n$ always holds for any state $\vrho$ since $\<\mA_k, \vrho \> \le \|\vrho\| \trace(\mA_k)\le \trace(\mA_k) = \frac{1}{d^n}$ for all $k$. However, we note that the above pathological case with spiky probability distribution may happen rarely in practice. For instance, consider a random density matrix $\vrho = \vu {\vu}^\dagger$, where\footnote{Similar result also holds when $\vu$ is randomly generated from the unit sphere according to the Haar measure.} each entry in $\vu \in\C^{d^n}$ in generated i.i.d. from a complex normal distribution $\calC\calN(0, \frac{1}{d^{n}})$. In this case, each term $\< \mA_k, \vu {\vu}^\dagger \>$ has mean $\E[\< \mA_k, \vu {\vu}^\dagger \>] = \frac{\trace(\mA_k)}{d^n} = \frac{1}{d^{2n}}$. Moreover, since $\< \mA_k, \vu {\vu}^\dagger \>$ is a second-order polynomial in the entries of Gaussian random vector, it is a subexponential random variable \cite[Proposition 2.4]{zajkowski2020bounds} with subexponential norm \footnote{For a random variable $X$, its subgaussian norm and subexponential norm are respectively defined as $ \|X\|_{\psi_2} = \inf \{t>0, \ \E e^{\frac{|X|^2}{t^2}} \leq 2  \}$ and $\|X\|_{\psi_1} = \inf \{t>0, \ \E e^{\frac{|X|}{t}} \leq 2  \}$.} $\|\< \mA_k, \vu {\vu}^\dagger \> - \frac{1}{d^{2n}}\|_{\psi_1}\leq \frac{c}{d^n} \|\sqrt{d^n}\vu \|_{\psi_2}^2\|\mA_k\|_F + \frac{1}{d^{2n}}\lesssim \frac{1}{d^{2n}}$ for some constant $c$.
We can then invoke
concentration inequality for subexponential random variable \cite{vershynin2018high} to obtain that $\max_{k} \<\mA_k, \vrho^\star \> \lesssim \frac{n }{d^{2n}}$ holds with probability at least $1 - e^{-Cn}$ for some positive constant $C$. In other words, for such a random density matrix $\vrho$, with high probability the probability distribution under SIC-POVM is quite uniform as $\gamma_{\max}(\vrho) \lesssim n$.

We note that the dependence of $\gamma_{\max}(\vrho^\star)$ may arise from technical difficulties, and could be eliminated by more sophisticated analysis, which is left as future work. Instead, we exploit additional properties in the general $t$-design to improve the bound.

\paragraph*{Case $(ii)$: general $t$-designs} \ \ To provide a uniform bound for all MPOs without the dependence on $\gamma_{\max}(\vrho^\star)$ and $\gamma_{\min}(\vrho^\star)$, we can consider general $t$-designs, utilizing the property that uniform sampling from a general $t$-design mimics the first $2t$ moments of sampling uniformly according to the Haar measure. The following result establishes an improved guarantee for estimating MPO with general  $t$-designs.
\begin{theorem}[Stable recovery with general $t$-designs]
\label{Recovery error of various t-designs POVM 3}
Suppose $\{\mA_1, \ldots, \mA_{K}\}$ form a set of  general $t$-design POVMs with
\begin{eqnarray}
\label{requirement of t designs in the main paper 1}
t\ \geq\ 6+\frac{2\,\theta\beta}{\log\frac{1}{4\theta}},
\end{eqnarray}
where
\begin{eqnarray}
\label{requirement of t designs in the main paper 11}
\theta \lesssim  \sqrt{n d^2 \ol{r}^2  \log n / M},
\end{eqnarray}
and $\beta$ is defined in \eqref{requirement of t designs in the main paper 12}. Let $\vrho^\star\in\C^{d^n\times d^n}$ be an MPO state with MPO ranks $(r_1,\ldots,r_{n-1})$.  If we measure the state $M$ times using the POVM, then with probability $1 - e^{- \Omega( n d^2 \ol{r}^2\log n) }$, the solution $\wh \vrho$ of the constrained least squares \eqref{The loss function in QST for SIC-POVM} satisfies
\begin{eqnarray}
\label{final conclusion of recovery error1 for t-design larger 3}
\|\wh{\vrho} - \vrho^\star\|_F\lesssim \sqrt{\frac{n d^2 \ol{r}^2  \log n}{M}},
\end{eqnarray}
provided that $M = \Omega\big(  nd^2 \ol{r}^2 \log n  \big)$.
\end{theorem}
The proof is given in Appendix~\ref{Proof of recovery error of t-designs POVM 3 in Appe}. With general spherical $t$-design POVMs,
\Cref{final conclusion of recovery error1 for t-design larger 3} ensures a stable recovery of the ground-truth state with $\epsilon$-accuracy, provided that the number of state copies satisfies $M= O(n d^2 \ol{r}^2 \log n/\epsilon^2)$. Notably, this uniform guarantee is applied to all MPOs. Since $\theta\lesssim\sqrt{nd^2\ol{r}^2\log n/M}$, the above sample complexity implies $\theta=O(\epsilon)$. Consequently, the required design order satisfies $t\geq 6+O(\frac{\epsilon \beta}{\log(1/\epsilon)})$. Since $t$ is an integer and the second term in \eqref{requirement of t designs in the main paper 1} is positive but vanishes as $\epsilon$ decreases, the condition \eqref{requirement of t designs in the main paper 1} is satisfied by $t=7$ whenever $2\theta\beta\leq\log\frac{1}{4\theta}$, which holds for sufficiently small $\epsilon$. Hence, a $7$-design is sufficient for sufficiently high recovery accuracy.

To further characterize the possible range of $\beta$, for any fixed $\vrho\in\ol{\setX}_{4\vr+1}$, the second-moment identity of a $t$-design with $t\geq2$ in \eqref{The l2 norm of A(rho) approximate 2_designs} gives $\sum_{k=1}^{K}\langle\mA_k,\vrho\rangle^2 = \frac{d^n}{K(d^n+1)}$. Since the maximum of $K$ nonnegative numbers is no smaller than their average, we have $\beta \geq \sqrt{\frac{d^n}{d^n+1}}$. Hence, $\beta$ is always at least of constant order, with the lower bound approaching $1$ as $n$ increases. It is important to emphasize that the requirement on $t$ is primarily technical: higher-order moments of the $t$-design are needed to verify the condition in \eqref{requirement in the concentration inequality main paper}, which is required to apply \Cref{General bound of multinomial distribution Q cases1}. This naturally raises the question of whether this design-order requirement can be removed by considering a measurement ensemble with moments of all orders.

While the solution $\wh{\vrho}$ of \eqref{The loss function in QST for SIC-POVM} may be non-physical, we can impose additional PSD constraint to obtain a physical state without compromising the recovery guarantee in \Cref{Recovery error of various quantum POVM}. Alternatively, we can simply project $\wh \vrho$ onto the set of physical states $\setS_+:=\{\vrho \in \C^{d^n\times d^n}: \vrho \succeq \vzero, \trace(\vrho) = 1\}$. Denote by $P_{\setS_+}$ the projection onto the set $\setS_+$, which can be efficiently computed by projecting the eigenvalues onto a simplex \cite{condat2016fast}. Since the set $\setS_+$ is convex, the corresponding projector is non-expansive, and hence
\begin{eqnarray}
    \label{The nonexpansiveness of physical state}
    \|P_{\setS_+}(\wh{\vrho} )-\vrho^\star\|_F &\!\!\!\!=\!\!\!\!& \|P_{\setS_+}(\wh{\vrho} )-P_{\setS_+}(\vrho^\star)\|_F\nonumber\\
    &\!\!\!\!\le\!\!\!\!& \|\wh{\vrho}-\vrho^\star\|_F \le \epsilon,
\end{eqnarray}
which implies that the projection step ensures that the state becomes physically valid while preserving or even improving the recovery guarantee.

\paragraph*{Trace distance} \ \ The recovery guarantees discussed in this paper are established in the Frobenius norm (or Hilbert-Schmidt distance), which, while not a standard metric in quantum information, is widely adopted for analyzing MPO states and tensor networks \cite{baumgratz2013scalable,baumgratz2013scalablemax,guo2022quantum}. This measure is particularly useful for understanding how errors scale with factors such as the sample size and the number of state copies in large-scale MPO states. Notably, although the size of the density matrix grows exponentially with $n$, the computational complexity of the Frobenius norm for an MPO state remains linear with respect to $n$ in \cite{Oseledets11}. In contrast, there is no known efficient method to compute the trace distance or fidelity between two MPO states. Therefore, the choice of using Frobenius norm to characterize the recovery errors for MPO states is natural from a computational perspective.

On the other hand, an issue with the Frobenius norm is that the state itself can potentially have an exponentially small Frobenius norm, particularly for high-rank states. However, this is not typically the case for many physical systems, such as those at low temperatures or near pure states. For example, the thermal/Gibbs state, $\vrho^\star = \frac{e^{-\mH/T}}{\trace(e^{-\mH/T})}$, where $\mH$ denotes the Hamiltonian and $T$ the temperature (ignoring the Boltzmann constant) can be well approximated by a low-rank density matrix if $T$ is small. In addition, when $\mH$ describes a one-dimensional spin model with short-range interactions, the resulting thermal state can often be approximated by an MPO. Furthermore, even for pure states with MPS representations (a special case of MPO states), there exist no rigorous sample complexity bounds that are efficient (i.e. the number of state copies is proportional to the number of independent parameters encoding the state).

To present a formal theorem formulated in terms of the trace distance for practical quantum states that are appropriately low-rank, we utilize the result in \cite{coles2019strong} to establish the connection between trace distance and the one based on Frobenius distance. In particular, suppose $\vrho^\star$ is approximately low-rank and can be expressed as $\vrho^\star = \vrho^\star_s + \mN$, where $\vrho^\star_s$ is a rank-$s$ matrix\footnote{Here $s$ denotes the matrix rank, while $\vr = (r_1,\dots, r_{n-1})$ represents MPO ranks.} and $\mN$ represents the small residual. Note that $\mN =\vzero$ when the state is strictly low-rank. 
\begin{lemma}
\label{distance between F norm and trace norm lemma}
Suppose  $\vrho^\star = \vrho^\star_s + \mN$, where $\vrho^\star_s$ is a rank-$s$ matrix. Then for any $\wh\vrho\in\C^{d^n \times d^n}$, we have
\begin{eqnarray*}
\label{distance between F norm and trace norm}
\|\wh\vrho - \vrho^\star\|_1\leq 2\sqrt{s} \| \wh\vrho - \vrho^\star\|_F  + 2\sqrt{s}\|\mN\|_F + \|\mN\|_1.
\end{eqnarray*}
\end{lemma}
\begin{proof}
\begin{align*}
\|\wh\vrho - \vrho^\star\|_1 & \leq \|\wh\vrho - \vrho^\star_s\|_1 + \|\mN\|_1\nonumber\\
&\leq 2\sqrt{s}\| \wh\vrho - \vrho^\star_s \|_F + \|\mN\|_1\nonumber\\
& = 2\sqrt{s} (\| \wh\vrho - \vrho^\star + \mN \|_F) + \|\mN\|_1\nonumber\\
&\le 2\sqrt{s} \| \wh\vrho - \vrho^\star\|_F  + 2\sqrt{s}\|\mN\|_F + \|\mN\|_1,
\end{align*}
where the second inequality follows from \cite[Theorem 1]{coles2019strong}.
\end{proof}
\Cref{distance between F norm and trace norm lemma} supports the utilization of Hilbert-Schmidt distance for practical quantum states, demonstrating that the trace distance can be upper bounded by the Hilbert-Schmidt distance. This property (with $\mN = \vzero$) is utilized in \cite{haah2017sample} for obtaining guarantees in trace distance for low-rank states.
Now, based on \Cref{distance between F norm and trace norm lemma}, we can extend \Cref{Recovery error of various quantum POVM,Recovery error of various t-designs POVM 3} to trace distance. To unify results for both $2$-designs and general $t$-designs, we define $\gamma_t(\vrho^\star) = \begin{cases}
    \gamma_{\max}(\vrho^\star), & t = 2,\\
    1, & t \ \textup{satisfying \eqref{requirement of t designs in the main paper 1}} .
  \end{cases}$
\begin{cor}
Suppose $\{\mA_1, \ldots, \mA_{K}\}$ form a set of   $t$-design POVMs. Consider an MPO $\vrho^\star\in\C^{d^n \times d^n}$ that has MPO ranks $(r_1,\dots, r_{n-1})$ and is approximately low-rank, i.e, $\vrho^\star = \vrho^\star_s + \mN$.  If we measure the state $M$ times using the POVM under the settings of \Cref{Recovery error of various quantum POVM,Recovery error of various t-designs POVM 3}, then with high probability, the solution $\wh \vrho$ of the constrained least squares \eqref{The loss function in QST for SIC-POVM} satisfies
\begin{eqnarray}
\label{final conclusion of recovery error1 for t-design larger 3 trace norm}
\|\wh\vrho - \vrho^\star\|_1&\!\!\!\!\lesssim \frac{ d\ol{r}\sqrt{ ns \gamma_t(\vrho^\star) \log n} }{ \sqrt{ M}}\nonumber\\
&\!\!\!\!+ 2\sqrt{s}\|\mN\|_F +\|\mN\|_1.
\end{eqnarray}
\label{Recovery error of various t-designs POVM trace norm}
\end{cor}
As shown in {Corollary} \ref{Recovery error of various t-designs POVM trace norm}, when the ground truth MPO state is approximately low-rank, the recovery error in terms of trace distance also exhibits a polynomial growth with respect to $n$. Moreover, according to~\eqref{final conclusion of recovery error1 for t-design larger 3 trace norm}, this error can be controlled using a number $M$ of state copies that scales only polynomially in $n$. In contrast, low-rank quantum state tomography based on unstructured low-rank density matrices \cite{gross2010quantum,o2016efficient,haah2017sample,francca2021fast} typically requires $\Omega(d^n s)$ state copies. By exploiting the MPO structure, our approach reduces this exponential dependence on $n$ to a polynomial one. Beyond sample complexity, our proof techniques also differ from those for low-rank QST. The analysis in \cite{gross2010quantum} establishes recovery for a convex nuclear-norm minimization program by constructing a dual certificate via the golfing scheme, while \cite{haah2017sample} achieves the optimal sample complexity using collective (entangled) measurements across all copies, analyzed via the representation theory of the unitary and symmetric groups. Neither approach directly extends to MPO states: there is no tractable convex surrogate for the MPO ranks, and the set of MPOs is not invariant under global unitary transformations, which is essential for the analysis in \cite{haah2017sample}. Instead, we directly analyze the nonconvex constrained least-squares estimator \eqref{The loss function in QST for SIC-POVM} over the set of MPOs, where the exact isometry of $t$-design POVMs is combined with the multinomial concentration inequality and a covering argument whose complexity scales with the degrees of freedom of MPOs rather than the ambient dimension $d^n$.

\paragraph*{Lower bound}
To further characterize the fundamental limits on the number of copies required in MPO-based QST under the trace distance, we derive the following information-theoretic lower bound.
\begin{theorem}[Sample Complexity Lower Bound for MPO-based QST]
\label{thm:mpo_lower_bound}
Let $\vrho^\star$ be a strictly positive definite density matrix lying in the interior of the physical manifold $\calP$ of Matrix Product Operators (MPOs) on $\nqbit$ qudits with bond dimension $\overline{r}$. Let $\wh{\vrho}$ be an efficient, unbiased estimator of $\vrho^\star$ on $\calP$ constructed from $M$ independent quantum measurements. When $M$ is sufficiently large, achieving $\E[\|\vrho^\star - \wh{\vrho} \|_1] \le \zeta$ requires
\begin{eqnarray}
\label{necessary condition for MPO-based QST}
    M = \Omega\bigg( \frac{\nqbit d^2 \overline{r}^2}{\zeta} \bigg).
\end{eqnarray}
\end{theorem}
The proof is provided in Appendix~\ref{proof: lower bound sample complexity}. \Cref{thm:mpo_lower_bound} shows that any unbiased estimator achieving a trace-norm accuracy $\zeta$ must use at least $\Omega(\nqbit d^2 \overline{r}^2 / \zeta)$ copies of the quantum state.
Finally, the result can be extended to biased estimators as well.

Comparing \eqref{necessary condition for MPO-based QST} with \eqref{final conclusion of recovery error1 for t-design larger 3 trace norm}, which requires $M = O(nsd^2\ol r^2\gamma_t(\vrho^\star)\log n/\epsilon^2)$ state copies to achieve $\|\wh\vrho-\vrho^\star\|_1\le\epsilon$ when $\mN=\vzero$, we see that the linear dependence on the number of degrees of freedom $nd^2\ol r^2$ in the upper bound is necessary and cannot be avoided by algorithmic improvements. However, the above upper and lower bounds do not yet provide a fully matched characterization. First, \eqref{final conclusion of recovery error1 for t-design larger 3 trace norm} implies a $1/\epsilon^2$ dependence on the accuracy, whereas \eqref{necessary condition for MPO-based QST} gives a $1/\zeta$ dependence. Second, the two bounds are established under different assumptions: the upper bounds are established in the Frobenius norm for general MPO states, and the trace-norm upper bound requires an additional approximate low-rank assumption, whereas \Cref{thm:mpo_lower_bound} is formulated in expected trace norm under the assumptions of strict positivity of $\vrho^\star$ and unbiasedness of the estimator, which in particular excludes low-rank states such as the case $s=1$ and $\mN=\vzero$. Closing this gap will require a more refined analysis that accounts for the dependence of the Frobenius-norm error on the scale of the target state, characterized by $\|\vrho^\star\|_F$, a lower bound that also applies to low-rank states, as well as sharper statistical properties of the measurement POVMs. These aspects provide promising directions for future work toward a more complete characterization of the optimal sample complexity for MPO-based quantum state tomography.

\section{Optimization Methods for MPO-based QST}
\label{sec: algorithm}
In the literature of tensor recovery for tensor train (TT) format tensors, several iterative algorithms \cite{Rauhut17, rauhut2015tensor,budzinskiy2021tensor,qin2024guaranteed,qin2024computational} have been proposed and analyzed from random measurements. While these algorithms can be applied to MPO by rewriting MPO as TT \footnote{To establish the equivalence, we first reshape $\vrho$ into an $n$-th order tensor $\calX$ of size $d^2\times d^2\times \cdots \times d^2$, where each pair $(i_\ell,j_\ell)$ is mapped to a single index $s_\ell=i_\ell+d(j_\ell-1)$ for $\ell=1,\ldots,n$. Consequently, the $(s_1,\ldots,s_n)$-th element of $\calX$ becomes
\begin{eqnarray*}
\calX(s_1,\dots,s_n) = \vrho(i_1 \cdots i_\nqbit, j_1\cdots j_\nqbit)
= \mX_1^{s_1} \mX_2^{s_2} \cdots \mX_\nqbit^{s_\nqbit},
\end{eqnarray*}
where with abuse of notation we denote by $\mX_\ell^{s_\ell} = \mX_\ell^{i_\ell,j_\ell}$.}, the analysis has predominantly focused on real-valued tensors with either clear measurements or corrupted by Gaussian noise. Thus, the guarantees cannot be directly applied to QST involving complex-valued MPOs and statistical noise induced by the empirical measurements that follow a multinomial distribution. In this section, we develop specialized methods with guaranteed performance for MPO-based QST.

\paragraph*{Spectral method} \ \
Spectral methods are widely used for signal estimation in nonconvex problems, particularly as initialization schemes based on the leading components of an appropriately constructed surrogate matrix~\cite{candes2015phase,luo2019optimal,lu2020phase,Ma21TSP,tong2021accelerating,Han20,TongTensor21,qin2024guaranteed,qin2025robust,qin2025scalable}. In our setting, this principle leads to a simple yet effective one-step estimator, which we refer to as a {\it spectral method}. Specifically, we construct
\begin{eqnarray}
\label{adjoint term of measurement A}
\mY := \sum_{k=1}^{K}\frac{K(d^n+1)}{d^n}\wh p_k \mA_k,
\end{eqnarray}
which corresponds to applying the adjoint operator of $\calA$ to the empirical measurements $\wh \vp$, with a scaling factor $K(d^n+1)/d^n$ that balances the energy according to~\eqref{The l2 norm of A(rho) approximate 2_designs}. The final estimate is then obtained by projecting $\mY$ onto the structured MPO set via $\calP_{\setX_{\vr}}$, i.e.,
\begin{eqnarray}
    \label{equation of spectral initialization}
    \wh\vrho = \calP_{\setX_{\vr}} ( \mY  ).
\end{eqnarray}
Since computing the optimal projection onto the set of TT format tensors is already known as NP-hard \cite{hillar2013most}, computing the optimal projection onto the set $\setX_{\vr}$ is challenging even ignoring the constraint on the trace. Here, we use a two-step approximate projection $\calP_{\setX_{\vr}}(\cdot ) = \calP_{\trace = 1}( \text{SVD}_{\vr}^{tt}(\cdot))$, where $\text{SVD}_{\vr}^{tt}(\cdot)$ denotes the TT-SVD operation \cite{Oseledets11} that can compute an MPO (may not have trace 1) approximation for a given density matrix \footnote{The output $\text{SVD}_{\vr}^{tt}(\mA)$ is always Hermitian if $\mA$ is Hermitian.}, and $\calP_{\trace = 1}(\cdot)$ represents the projection onto the convex set $\{\vrho\in\C^{d^n\times d^n}: \trace(\vrho) = 1 \}$ to have unit trace \footnote{This can be achieved by computing the eigendecomposition and subsequently projecting the eigenvalues onto the simplex \cite{condat2016fast}. While this process may increase the MPO ranks, our observations indicate that this increase is often minimal or does not occur. However, this process (particularly computing the eigendecomposition) could be computationally expensive for large quantum states.
Alternatively, in our implementation, we simply use $\calP_{\trace=1}(\vrho) = \vrho/\trace(\vrho)$ which is computationally efficient while simultaneously preserving the MPO structure. Thus, in the analysis, we will assume the projection $\calP_{\trace = 1}(\cdot)$ also preserves the MPO structure of the input.}. While the TT-SVD operation $\text{SVD}_{\vr}^{tt}(\cdot)$ may not produce an optimal approximation, it can be approved to be sub-optimal. In particular, for any density matrix $\vrho$, we have \cite{Oseledets11} $\|\text{SVD}_{\vr}^{tt}(\vrho)- \vrho||_F^2 \leq (n-1)\min_{\vrho' \in\setX_{\vr}}\|\vrho' - \vrho\|_F^2$.

To unify results for both $2$-designs and general $t$-designs satisfying \eqref{requirement of t designs in the main paper 1}, recall the definition  $\gamma_t(\vrho^\star) = \begin{cases}
    \gamma_{\max}(\vrho^\star), & t = 2,\\
    1, & t \ \textup{satisfying \eqref{requirement of t designs in the main paper 1}}.
  \end{cases}$
The following theorem establishes an upper bound on the recovery error of the spectral estimator $\wh\vrho$ given in \eqref{equation of spectral initialization}.
\begin{theorem}
\label{spectral initialization in GD for global SIC-POVM}
Suppose that $\{\mA_1, \ldots, \mA_{K}\}$ form a set of   $t$-design POVMs, which is used to measure an MPO state $\vrho^\star$ $M$ times under the settings of \Cref{Recovery error of various quantum POVM,Recovery error of various t-designs POVM 3}. By utilizing the spectral method in \eqref{equation of spectral initialization}, with probability $ 1 - e^{- \Omega( n d^2 \ol{r}^2\log n)}$, we have
\begin{eqnarray}
    \label{proof of upper bound spec main paper}
&&\hspace{-1.5cm}\|\wh\vrho - \vrho^\star\|_F\leq (1+ \sqrt{n-1})\cdot O\bigg(\frac{ d\ol r \sqrt{ n \gamma_t(\vrho^\star) \log n}}{ \sqrt{M}}   \bigg).
\end{eqnarray}
\end{theorem}
Its proof is given in Appendix~\ref{Proof of spectral initialization in the GD}. The scaling term $\sqrt{n-1}$ arises due to the sub-optimal bound for the projection $\calP_{\setX_{\vr}}(\cdot ) = \calP_{\trace = 1}( \text{SVD}_{\vr}^{tt}(\cdot))$. Recall that $\gamma_t(\vrho^\star) = 1$ for all $\vrho^\star$ when general spherical $t$-design POVMs satisfying \eqref{requirement of t designs in the main paper 1}, and $\gamma_t(\vrho^\star) = \gamma_{\max}(\vrho^\star)$ for spherical 2-design POVMs (i.e., $t=2$). \Cref{spectral initialization in GD for global SIC-POVM} shows that  when $\gamma_t(\vrho^\star)$ is small---true for all MPOs when using $t$-designs satisfying \eqref{requirement of t designs in the main paper 1} and for those having relatively uniformed probability distribution with 2-design POVMs---the spectral method provides a relatively good estimation with a polynomially large $M$. In particular,  \eqref{proof of upper bound spec main paper} ensures that $\|\wh\vrho - \vrho^\star\|_F^2$ scales in the order of $\frac{n^2d^2 \ol r^2}{M}$ when ignoring the log terms and $\gamma_t(\vrho^\star)$.
We note that the measurement requirements are slightly more stringent than those in \Cref{Recovery error of various quantum POVM,Recovery error of various t-designs POVM 3}, due to the sub-optimal projection bound for $\calP_{\setX_{\vr}}(\cdot)$, which introduces an additional factor of $n$. Nonetheless, the spectral method remains a simple and efficient one-step estimator.

\paragraph*{Projected Gradient Descent} \ \
Motivated by the sub-optimal error bound of the spectral method, we ask whether there exists an iterative optimization approach that can achieve the optimal recovery guarantees established in \Cref{Recovery error of various quantum POVM,Recovery error of various t-designs POVM 3}. To this end, we consider projected gradient descent (PGD) as a natural candidate for refining the estimate. We begin by reiterating the loss function in \eqref{The loss function in QST for SIC-POVM}:
\begin{eqnarray}
    \label{The loss function in QST for various quantum measurements}
    \min_{\vrho\in\setX_{\vr}} g(\vrho) = \|\calA(\vrho) - \wh\vp\|_2^2.
\end{eqnarray}
We solve this constrained optimization problem by a PGD that iteratively updates
\begin{eqnarray}
    \label{gradient descent in the TT format}
    \vrho^{(\tau+1)} = \calP_{\setX_{\vr}}(\vrho^{(\tau)} - \mu \nabla g(\vrho^{(\tau)})  ),
\end{eqnarray}
where $\mu$ is the step size, $\nabla g(\vrho^{(\tau)}) = \sum_{k=1}^{K} (\< \mA_k, \vrho^{(\tau)}  \> - \wh\vp) \mA_k$ is the Wirtinger gradient \footnote{Note that $g(\vrho)$ is non-holomorphic (i.e., not complex differentiable), which poses challenges for the development and analysis of standard gradient-based optimization algorithms. Fortunately, we can adopt Wirtinger gradient to extend the classical gradient concept to functions of complex variables \cite{zhang2017matrix}. In our context, we first parameterize the objective
function as $g(\vrho,\vrho^*) = \sum_{k=1}^K(\<\mA_k^*,\vrho^* \> - \wh p_K^*)(\<\mA_k,\vrho \> - \wh p_K)$, and then compute the Wirtinger gradient $\nabla_{\vrho^*} g(\vrho,{\vrho^*})$ by treating $\vrho$ and $\vrho^*$ as two independent variables. We can then use $\nabla_{\vrho^*} g(\vrho,{\vrho^*})$ as the gradient in standard gradient-based optimization algorithms. To simplify the notation, we denote the Wirtinger gradient $\nabla_{\vrho^*} g(\vrho,{\vrho^*})$ simply as  $\nabla g(\vrho)$.}, and $\calP_{\setX_{\vr}}(\cdot )$ is the projection onto the set $\setX_{\vr}$.

We now provide error analysis for the PGD in \eqref{gradient descent in the TT format}.
\begin{theorem}
\label{convergence rate in GD for global SIC-POVM}
Suppose that $\{\mA_1, \ldots, \mA_{K}\}$ form a set of $t$-design POVMs, which is used to measure an MPO state $\vrho^\star$ $M$ times under the settings of \Cref{Recovery error of various quantum POVM,Recovery error of various t-designs POVM 3}. Given an initialization $\vrho^{(0)}$, suppose the PGD in \eqref{gradient descent in the TT format} uses step size $\frac{bK(d^n+1)}{d^{n}(1+b) }<\mu\leq \frac{(d-1)(d^n+1) K}{  d^{n+1}}$ and that the projection $\calP_{\setX_{\vr}}$ satisfies
\begin{eqnarray}
    \label{requirement of projection}
    \hspace{-0.4cm}\|\calP_{\setX_{\vr}}(\vrho^{(\tau)}  ) - \vrho^\star \|_{F}^2  \leq  (1+ b)\|\vrho^{(\tau)}   - \vrho^\star \|_{F}^2, \nonumber\\ \hspace{-0.4cm}\text{with} \ b< d-1.
\end{eqnarray}
Then with probability at least $ 1 - e^{- \Omega( n d^2 \ol{r}^2\log n )}$, the PGD converges at a  linear rate and
\begin{eqnarray}
    \label{expansion of vrho - rvrho* final main theorem}
    \hspace{-0.8cm}\|\vrho^{(\tau)} - \vrho^\star\|_F^2 \leq O\bigg( \frac{n d^2\ol{r}^2 \gamma_t(\vrho^\star) \log n }{M }  \bigg)
\end{eqnarray}
for sufficiently large $\tau$.
\end{theorem}
The proof is given in Appendix~\ref{Proof of convergence rate in the GD}. \Cref{convergence rate in GD for global SIC-POVM} provides a linear convergence analysis of the PGD algorithm along with its error characterization. Moreover, the recovery error presented in \eqref{expansion of vrho - rvrho* final main theorem} aligns closely with the results in \Cref{Recovery error of various quantum POVM,Recovery error of various t-designs POVM 3}. We remark that the projection condition in~\eqref{requirement of projection} is invoked solely to establish the linear convergence rate and is not required for the error bound once convergence has been achieved. Consequently, the recovery guarantee in~\eqref{expansion of vrho - rvrho* final main theorem} remains valid even when this condition is not strictly enforced. In practice, we therefore adopt the two-step approximate projection $\calP_{\setX_{\vr}}(\cdot ) = \calP_{\trace = 1}( \text{SVD}_{\vr}^{tt}(\cdot))$ which is also used in the spectral method.

\paragraph*{Projected stochastic gradient descent} \ \ The primary computational complexity of the PGD \eqref{gradient descent in the TT format} arises from computing the gradient, which involves an exponentially large number ($K$) of inner product operations between $\mA_k$ and $\vrho^{(\tau)}$, with a complexity of $O(K d^{2n})$ by naive computation. When the POVMs $\{\mA_k\}_{k=1}^K$ also have MPO representations, such as the local measurements as detailed in the subsequent section, the complexity of each inner product can be reduced to $O(n{\ol r}^3)$ from $O(d^{2n})$ by using the dot product method \cite{Oseledets11} for two MPOs. Similarly, the gradient update $\vrho^{(\tau)} - \mu \nabla g(\vrho^{(\tau)})$ in \eqref{gradient descent in the TT format} can also be represented in the MPO form rather than as a full density matrix, allowing us to utilize TT-rounding to efficiently compute TT-SVD with polynomial computational complexity in $n$ \cite{Oseledets11}.

The issue of computing all the $K$ inner products in one iteration can be addressed by using stochastic methods \cite{bottou2018optimization}. In our context, exploiting the fact that the empirical probabilities $\{\wh{p}_i\}$ are very sparse (at most $M$ of them are nonzero) when the measured time $M$ is much smaller than $K$, we can utilize a customized stochastic method. Specifically, unlike standard approaches that either randomly or sequentially select a subset of samples for computing gradient in each iteration, we can use the nonzero measurements more frequently. For instance, each time (often called epoch in machine learning), we first select a subset (say $N$) of empirical probabilities that include all the nonzero ones and a few numbers of zero ones (which can be either randomly selected or selected in a sequential), and then perform $N/B$ iterations of PSGD with each iteration using $B$ measurements selected sequentially form this subset.

\section{Numerical Experiments}
\label{sec: experimental results}

To our best knowledge, there are no known efficient implementations of SIC-POVM or spherical $t$-designs with $t\ge 2$ for a many-qubit system using local quantum circuits available in current or near-future quantum hardware. It is not surprising that the sample-efficient measurement setting for recovering an MPO state is not efficiently implementable, as this is also the case for the QST of generic quantum states \cite{haah2017sample,guctua2020fast,francca2021fast}. To make our results more experimentally relevant, in this section we consider a different type of IC-POVM that is efficiently implementable. Specifically, we consider the following local SIC-POVM $\{\mB_i\}_{i=1}^4$ for each qubit
\begin{eqnarray} \label{set of local SIC-POVM}
    \hspace{-0.6cm}\bigg\{\begin{bmatrix}\frac{1}{2}& 0 \\ 0 & 0 \end{bmatrix}, \begin{bmatrix}\frac{1}{6}& \frac{\sqrt{2}}{6} \\ \frac{\sqrt{2}}{6} & \frac{1}{6} \end{bmatrix}, \begin{bmatrix}\frac{1}{6}& \frac{\sqrt{2}}{6}e^{-i\frac{2\pi}{3}} \\ \frac{\sqrt{2}}{6}e^{i\frac{2\pi}{3}} & \frac{1}{6} \end{bmatrix}, \nonumber\\ \hspace{-0.6cm}\begin{bmatrix}\frac{1}{6}& \frac{\sqrt{2}}{6}e^{-i\frac{4\pi}{3}} \\ \frac{\sqrt{2}}{6}e^{i\frac{4\pi}{3}} & \frac{1}{6} \end{bmatrix}  \bigg\}.
\end{eqnarray}
We can then generate the set of IC-POVM for an $n$-qubit system as $\{\mA_i \}_{i=1}^{4^n} = \{\mB_{i_1}\otimes \cdots \otimes \mB_{i_n} \}_{i_1,\dots,i_n}$. Importantly, such IC-POVM has been recently demonstrated experimentally using trapped ions in a scalable way \cite{stricker2022experimental} and can be implemented in most quantum computing hardware platforms using either ancilla qubits or extra low-energy states in the physical carriers of the qubits.

\begin{figure*}[t]
\centering
\subfigure[]{
\begin{minipage}[t]{0.22\textwidth}
\centering
\includegraphics[width=\textwidth]{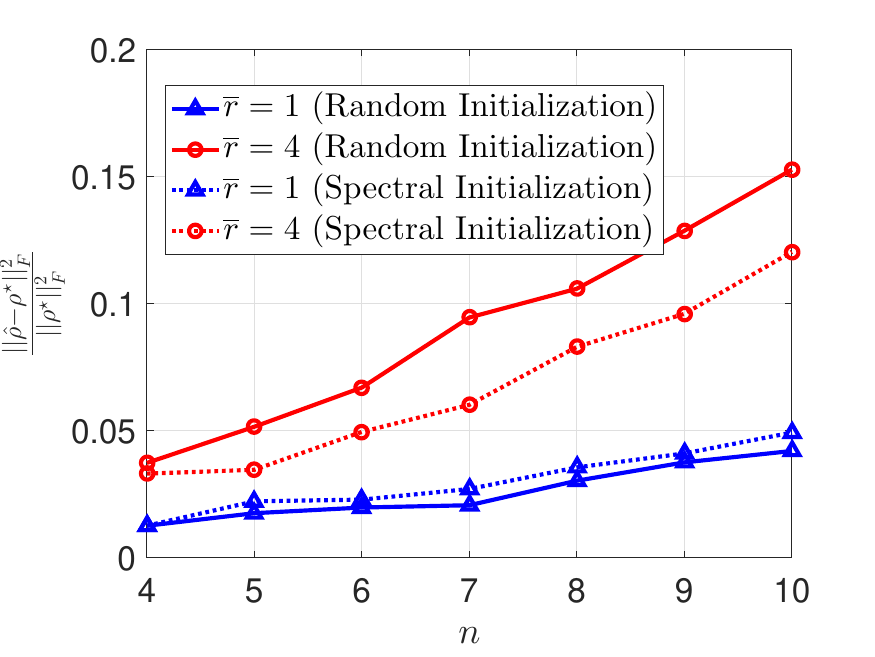}
\end{minipage}
\label{Performance Comparison_initilaizaiton recovery error}
}
\subfigure[]{
\begin{minipage}[t]{0.22\textwidth}
\centering
\includegraphics[width=\textwidth]{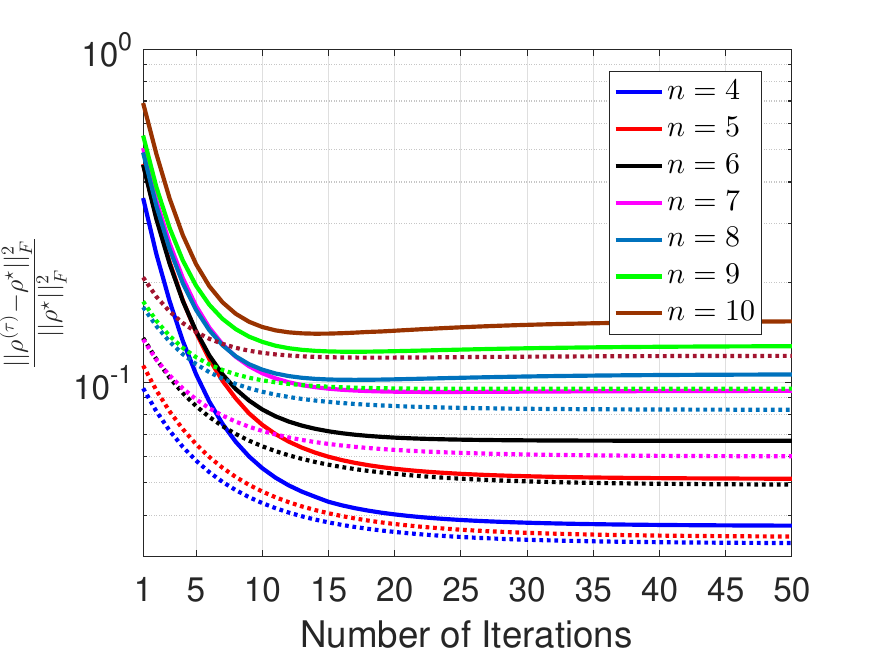}
\end{minipage}
\label{Performance Comparison_convergence}
}
\subfigure[]{
\begin{minipage}[t]{0.22\textwidth}
\centering
\includegraphics[width=\textwidth]{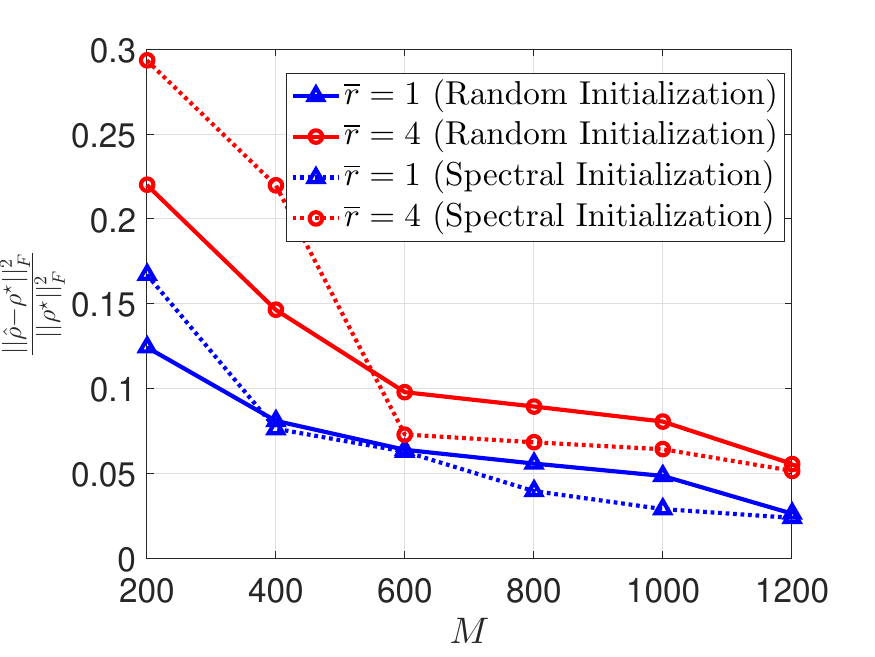}
\end{minipage}
\label{Performance Comparison_different M}
}
\subfigure[]{
\begin{minipage}[t]{0.22\textwidth}
\centering
\includegraphics[width=\textwidth]{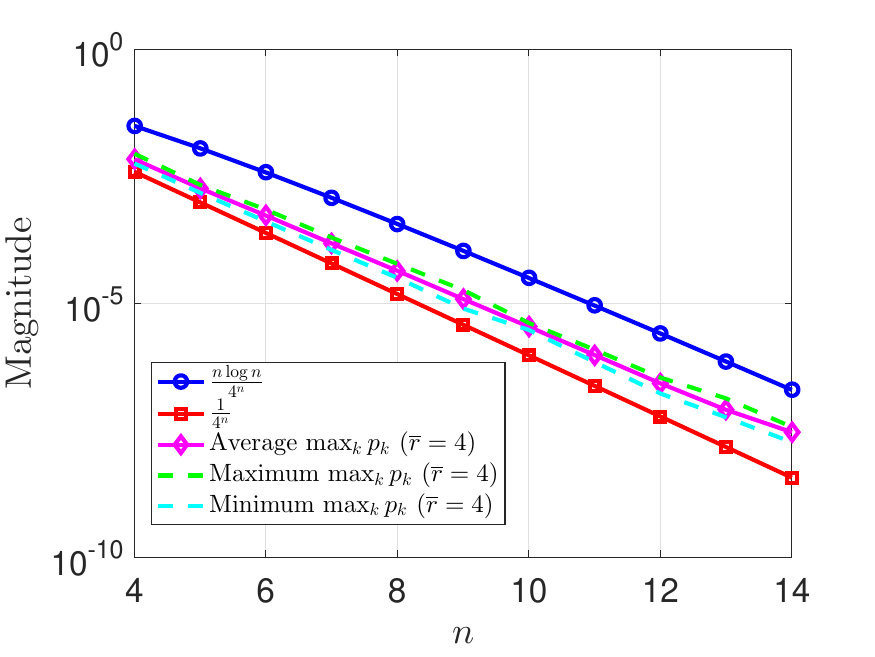}
\end{minipage}
\label{probability}
}
\caption{Illustration of (a) recovery error for different initialization methods, $\ol r$ and $n$ with $M = 3000$; (b) convergence rate of PGD for different $n$, spectral initialization (dotted lines), and random initialization (solid lines) with $M = 3000$ and $\ol r=4$; (c) recovery error with $n=5$ for different initialization methods, $\ol r$ and $M$; (d) $p_\text{max}$ for the MPDO.}
\end{figure*}

While our main results (\Cref{Recovery error of various quantum POVM,Recovery error of various t-designs POVM 3}) do not apply to such a measurement setting, we show that certain MPO states can be recovered using this IC-POVM using a number of state copies only linear in $n$ and the PGD and PSGD algorithms developed in \Cref{sec:stable recovery}. A more systematic study on what types of MPO states can be recovered efficiently using the local IC-POVM can be found in \cite{Jameson2024}, using a different optimization algorithm based on maximal likelihood estimation (MLE). As an example, here we consider a random matrix product density operator (MPDO) state \cite{verstraete2004matrix} of the following form \eqref{DefOfMPO}, i.e, the $(i_1 \cdots i_\nqbit, j_1 \cdots j_\nqbit)$-th entry of $\vrho^\star$ can be generated by
\begin{eqnarray}
\vrho^\star(i_1 \cdots i_\nqbit, j_1 \cdots j_\nqbit)
=  {\mX_1^\star}^{i_1,j_1} {\mX_2^\star}^{i_2,j_2} \cdots {\mX_\nqbit^\star}^{i_\nqbit,j_\nqbit}.\nonumber
\end{eqnarray}
To ensure $\vrho^\star$ PSD, we generate each matrix ${\mX_\ell^\star}^{i_\ell,j_\ell}$ as \cite{verstraete2004matrix}
\begin{eqnarray}
\label{form of X l}
{\mX_\ell^\star}^{i_\ell,j_\ell} = \sum_{a_\ell = 1}^{K_\ell} {\mA_\ell^\star}^{i_\ell, a_\ell}\otimes ({\mA_\ell^\star}^{j_\ell, a_\ell})^*,
\end{eqnarray}
where the symbol $*$ represents the complex conjugate, $\otimes$ denotes the Kronecker product, $K_\ell$ is an arbitrary positive integer that controls purity of the generated state,
and each matrix ${\mA_\ell^\star}^{i_\ell, a_\ell}$ has dimension $\kappa\times \kappa$, except for ${\mA_1^\star}^{i_1, a_1}$ and ${\mA_n^\star}^{i_n, a_n}$, which are of dimension $1\times\kappa$ and $\kappa\times 1$ respectively. Each entry of ${\mA_\ell^\star}^{i_\ell, a_\ell}$ is i.i.d. randomly generated with both its real and imaginary parts drawn uniformly from the interval $[-1,1]$. In the experiments, we set $K_\ell = 10$ to ensure the state is sufficiently mixed; using $K_\ell = 1$ would generate a pure state. The generated $\vrho^\star$ is an MPO with a bond dimension of $\ol r = \kappa^2$ and is always PSD. To be a physical state, $\vrho^\star$ also needs to have unit trace, which can be done by calculating
\begin{eqnarray}
\trace(\vrho^\star)
= ({\mX_1^\star}^{1,1} + {\mX_1^\star}^{2,2})\cdots ({\mX_n^\star}^{1,1} + {\mX_n^\star}^{2,2}),
\nonumber\end{eqnarray}
and then dividing each matrix ${\mX_\ell^\star}^{i_\ell,j_\ell}$ by $(\trace(\vrho^\star))^{\frac{1}{n}}$.

In order to validate the theoretical analysis presented in the preceding sections, we apply PGD and PSGD, both initialized via a spectral method, hereafter referred to as spectral initialization, across a range of measurement settings. To further mitigate the high computational complexity of the spectral method, we will also use a random initialization generated by a random MPDO as the starting point for both PGD and PSGD. This approach is motivated by our observation that random initialization, when the target is a random MPDO, can effectively serve as a convergence-guaranteeing initialization.

\paragraph*{Performance of PGD}
\ \
In the first set of experiments, we set $M = 3000$ for empirical measurements. Since we will perform experiments for different number of qubits, we employ PGD with diminishing step sizes $\mu_\tau = \mu_0 \cdot 2^n \cdot \lambda^\tau$ (with $\lambda = 0.9$), which is found to provide stable performance across varying $n$. When a random initialization is used, we set $\mu_0 = \frac{5}{4}$ for $\ol r=1$ and $\mu_0 = \frac{5}{8}$ for $\ol r=4$. As for spectral initialization, $\mu_0$ is set to $\frac{5}{8}$ for $\ol r=1$ and $\frac{5}{16}$ for $\ol r=4$.
\Cref{Performance Comparison_initilaizaiton recovery error} depicts the recovery error as a function of $n$ for various $\ol r$ and initialization methods. As anticipated, the recovery error tends to rise with an increase in $\ol r$, yet the PGD algorithm displays stability across all cases. Notably, we observe that the recovery error increases only polynomially rather than exponentially with respect to $n$, aligning with the findings in \Cref{Recovery error of various quantum POVM,Recovery error of various t-designs POVM 3}. Furthermore, \Cref{Performance Comparison_convergence} illustrates the convergence behavior of PGD with both spectral and random initializations. As $n$ increases, the quality of initial values from spectral initialization deteriorates, aligning with the theoretical analysis presented in \Cref{spectral initialization in GD for global SIC-POVM}. Conversely, although the recovery error with random initialization is slightly worse than that of the spectral method, it offers a reliable starting point for stable recovery while avoiding the high computational complexity associated with the spectral method.
Subsequently, we evaluate the recovery performance across varying values of $M$, selecting $\mu_0 = \frac{5}{32}$ for random initialization and $\mu_0 = \frac{1}{16}$ for spectral initialization, with $\lambda = 1$. In \Cref{Performance Comparison_different M}, we observe that the recovery performance of PGD with spectral and random initializations is comparable, with improved recovery error as $M$ increases. Finally, \Cref{probability} reveals that the numerator of the average maximum probability is polynomial, ensuring the effectiveness of the assumption.

\begin{figure}[t]
\centering
\subfigure[]{
\begin{minipage}[t]{0.22\textwidth}
\centering
\includegraphics[width=\textwidth]{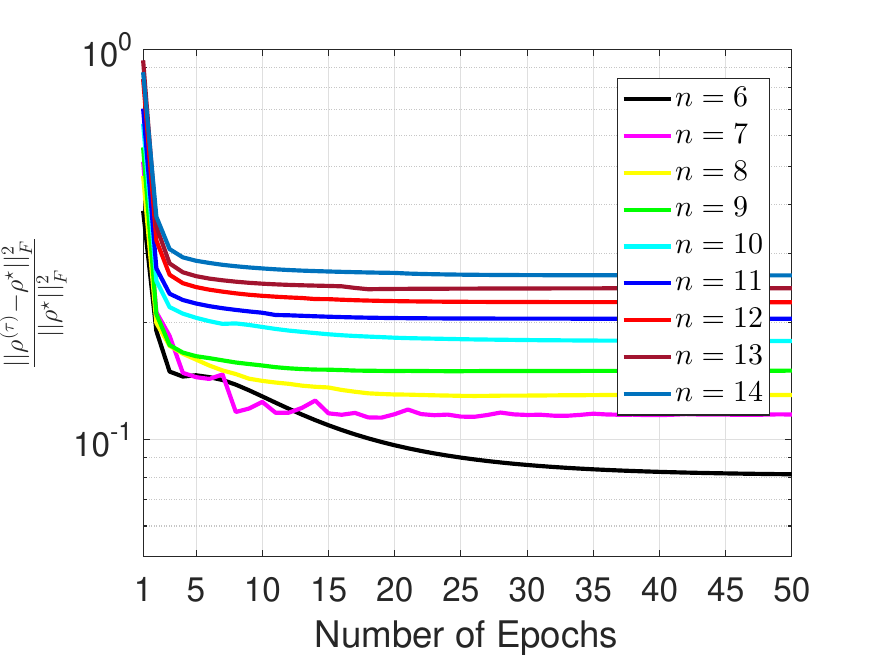}
\end{minipage}
\label{Convergence of SGD}
}
\subfigure[]{
\begin{minipage}[t]{0.22\textwidth}
\centering
\includegraphics[width=\textwidth]{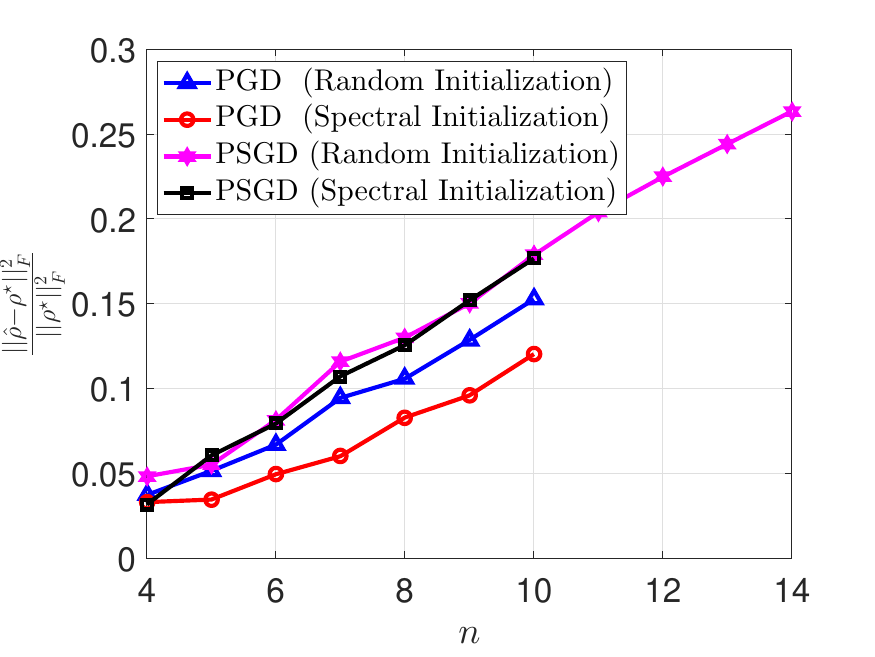}
\end{minipage}
\label{Accuracy of SGD}
}

\caption{(a) Convergence comparison of PSGD with random initialization for various $n$; (b) Performance comparison between PGD and PSGD with random and spectral initializations.}
\end{figure}

\paragraph*{Performance comparison between PGD and PSGD}\ \ In the second experiment, we compare the performance of PGD and PSGD in estimating the MPDO. We set $M = 3000$ and $\ol r = 4$.
We use the same setting for PGD in \Cref{Performance Comparison_initilaizaiton recovery error} and use diminishing step sizes $\mu_\tau = \frac{5}{4} \times 2^n \times 0.9^\tau$ for random initialization and $\mu_\tau = 10\times 0.9^\tau$ for spectral initialization. Additionally, we choose $N = 10 \times 4 n\ol r^2 = 640n$ measurements in one epoch and $B = 32$ in each iteration of one epoch. \Cref{Convergence of SGD} provides an analysis of the convergence properties of PSGD with random initialization, demonstrating that this achieves a linear convergence rate and exhibits stable recovery errors as $n$ increases. Furthermore, due to its polynomial computational complexity, PSGD with random initialization is viable for systems with up to 14 qubits. In \Cref{Accuracy of SGD}, we observe that PSGD with random initialization delivers performance comparable to that achieved with spectral initialization. Notably, although PSGD with polynomial measurements incurs a marginally higher recovery error compared to PGD, it nonetheless achieves a stable recovery error that grows polynomially with $n$.

To further elucidate the computational complexity of the two methods, we summarize the computational cost for a single iteration of PGD and one epoch of PSGD under local measurements as follows:
\begin{eqnarray*} \label{complexity of PGD}
    \text{PGD:} \ O\bigg( \underbrace{4^n n{\ol r}^3  }_{\text{inner product}} + \underbrace{ 4^n {\ol r}^2 }_{\text{TT-SVD}} + \underbrace{n{\ol r}^2 }_{\text{unit trace}}   \bigg),
\end{eqnarray*}
\begin{eqnarray*} \label{complexity of SPGD}
\text{PSGD:} \ O\bigg( \frac{N}{B}(\underbrace{B n{\ol r}^3  }_{\text{inner product}}
+  \underbrace{ n  ({\ol r} + B )^3 }_{\text{rounding}} + \underbrace{n{\ol r}^2 }_{\text{unit trace}} )  \bigg).
\end{eqnarray*}

\section{Conclusion}
\label{sec: conclusion}

This paper focuses on the sample complexity bounds for recovering structured quantum states in one-dimension that are represented as matrix product operators (MPOs) with finite bond dimensions. The study begins by establishing theoretical bounds on the accuracy of a constrained least-squares estimator for MPO state recovery, utilizing empirical measurements from a class of Informationally Complete Positive Operator-Valued Measure (IC-POVM). The results indicate that assuming probabilities measured by one POVM are approximately uniform for SIC-POVM or $2$-designs, or without any assumption for general $t$-designs satisfying \eqref{requirement of t designs in the main paper 1}, a stable recovery guarantee requires the number of state copies to be only linear in the the number of qudits $n$ and is thus sample efficient. We further establish a theoretical sample complexity lower bound for MPO reconstruction. Additionally, we introduce a spectral method and a PGD algorithm, both with provable recovery guarantees. These results address fundamental questions on whether quantum states with efficient classical representations can be recovered efficiently. While the measurement settings used in our sample complexity bounds are not efficiently implementable, we show that certain MPO states may be recovered efficiently using a local IC-POVM that can be experimentally achieved on current quantum hardware. It remains an open question what class of MPO states, or other structured quantum states, can be recovered not only using an efficient number of state copies, but also using an efficiently implementable measurement setting \cite{Jameson2024}.

\section*{Acknowledgment}
We acknowledge funding support from NSF Grants No.\ CCF-1839232, PHY-2112893, CCF-2106834, CCF-2241298 and ECCS-2409701 as well as the W. M. Keck Foundation. We thank the Ohio Supercomputer Center for providing the computational resources needed in carrying out this work.

\bibliographystyle{IEEEtran}
\bibliography{quantum}

\onecolumn
\appendices

\section{Proof of Lemma 1}
\label{Proof of RIP SIC POVM in Appe}

\begin{proof}
According to \cite{rastegin2014notes}, any Hermitian matrix $\vrho\in\C^{d^n\times d^n}$ can be represented in terms of elements of the dual basis as
\begin{eqnarray}
\label{The expression of density matrix in the dual basis}
\vrho = \sum_{k=1}^{d^{2n}}\<\mA_k, \vrho \>\widetilde{\mA}_k.
\end{eqnarray}
Here,  the dual basis $\{\widetilde{\mA}_k  \}$ is a basis such that $    \<\mA_k, \widetilde{\mA}_j \> = \begin{cases}
    1, & k = j\\
    0, & k \neq j
  \end{cases}$. For a general SIC-POVM $\{\mA_k  \}$, the dual basis is comprised by operators \cite{gour2014construction}
\begin{eqnarray}
\label{The defi of dual basis}
\widetilde{\mA}_k = d^n(d^n + 1) \mA_k - \mId_{d^n}.
\end{eqnarray}
Using $\vrho$ in \eqref{The expression of density matrix in the dual basis}, we have
\begin{eqnarray}
\label{The square of rho}
\trace(\vrho^2) &\!\!\!\! = \!\!\!\!& (d^n(d^n+1) - 1)\| \calA(\vrho)\|_2^2 - \sum_{j\neq k}\< \mA_j, \vrho\>\< \mA_k, \vrho\>\nonumber\\
&\!\!\!\! = \!\!\!\!&(d^n(d^n+1) - 1)\| \calA(\vrho)\|_2^2 - \big(\big(\sum_{i=1}^{d^{2n}}\<\mA_i, \vrho \>  \big)^2 - \| \calA(\vrho)\|_2^2\big) \nonumber\\
&\!\!\!\! = \!\!\!\!&d^n(d^n+1)\| \calA(\vrho)\|_2^2 - (\trace(\vrho))^2,
\end{eqnarray}
where the first line follows (22) in \cite[Proposition 1]{rastegin2014notes}. In addition, due to $\vrho = \vrho^\dagger$, we have $\trace(\vrho^2) = \trace(\vrho^\dagger \vrho) = \|\vrho\|_F^2$. This completes the proof.
\end{proof}

\section{Proof of Lemma 2}
\label{Proof of RIP 2 designs in Appe}

\begin{proof}
According to \cite[Theorem 1]{dall2014accessible} and \Cref{ty of 2-designs}, we have
\begin{eqnarray}
\label{sum of 2-designs measurement}
\int (\vw\vw^\dagger)^{\otimes 2} d\vw = \frac{\mId + \mS }{d^n(d^n+1)},
\end{eqnarray}
where $\mS$ is the swap operator.  Taking the inner product between \eqref{sum of 2-designs measurement} and $\frac{d^{2n}}{K}\vrho^{\otimes 2}$ we get
\begin{eqnarray}
\label{The l2 norm of A(rho) 2_designs proof}
\bigg\<\int (\vw\vw^\dagger)^{\otimes s} d\vw, \frac{d^{2n}}{K}\vrho^{\otimes 2} \bigg\> &\!\!\!\! = \!\!\!\!& \frac{d^n(\trace(\vrho^{\otimes 2}) + \trace(\mS\vrho^{\otimes 2}) )}{K(d^n+1)}\nonumber\\
&\!\!\!\! = \!\!\!\!&\frac{d^n((\trace(\vrho))^2 + \trace(\vrho^2)) }{K(d^n+1)}\nonumber\\
&\!\!\!\! = \!\!\!\!&\frac{d^n((\trace(\vrho))^2 + \|\vrho\|_F^2 )}{K(d^n+1)},
\end{eqnarray}
where the penultimate line follows $\trace(\mS\vrho^{\otimes 2}) = \trace(\vrho^2)$ \cite[Lemma 17]{KuengACHA17} for any Hermitian matrix $\vrho$.

Based on $\trace((\vw_k\vw_k^\dagger)^{\otimes 2}(\vrho^{\otimes 2})  ) = \trace( (\vw_k\vw_k^\dagger\vrho)\otimes (\vw_k\vw_k^\dagger\vrho) ) = (\trace(\vw_k\vw_k^\dagger\vrho))^2$, this completes the proof.
\end{proof}

\section{Proof of Theorem 1}
\label{Proof of recovery error of SIC-POVM in Appe}

\begin{proof}
Now, we start by upper-bounding $\<  \veta, \calA(\wh{\vrho} - \vrho^\star) \>$. Towards that goal, we first rewrite this term as
\begin{eqnarray}
    \label{upper bound of entire variable_SIC_POVM}
    \<  \veta, \calA(\wh{\vrho} - \vrho^\star) \> &\!\!\!\!=\!\!\!\!& \sum_{k=1}^{K} \eta_{k}\<\mA_k, (\wh{\vrho} - \vrho^\star) \>\nonumber\\
    &\!\!\!\! \leq\!\!\!\!& \|\wh{\vrho} - \vrho^\star\|_F \max_{\vrho\in \ol \setX_{2\vr}} \sum_{k=1}^{K} \eta_{k}\<\mA_k, \vrho \>,
\end{eqnarray}
where we denote by $\ol \setX_{\vr}$ the normalized set of MPOs with $\vr = [r_1,\dots, r_{n-1}]$:
\begin{eqnarray}
\label{SetOfMPO normalized}
\begin{split}
\ol \setX_{\vr}= \Big\{& \vrho\in\C^{d^n\times d^n}:\ \vrho = \vrho^\dagger, \|\vrho\|_F =  1, \trace(\vrho) = 0,  \vrho(i_1 \cdots i_\nqbit, j_1\cdots j_\nqbit) = \mX_1^{i_1,j_1} \mX_2^{i_2,j_2} \cdots \mX_\nqbit^{i_\nqbit,j_\nqbit}, \\
& \mX_\ell^{i_\ell,j_\ell}\in\C^{r_{\ell-1}\times r_\ell}, \ell\in[n-1], r_0=r_n=1\Big\}.
\end{split}
\end{eqnarray}
It is important to highlight that in comparison to $\setX_{\vr}$ in \eqref{SetOfMPO}, $\ol \setX_{\vr}$ additionally includes $\|\vrho\|_F = 1$ and $\trace(\vrho) = 0$.

The rest of the proof is to bound $\max_{\vrho\in \ol \setX_{2\vr}} \sum_{k=1}^{K} \eta_{k}\<\mA_k, \vrho \>$, which will be achieved by a covering argument. Following \cite[Appendix B-A]{qin2024quantum}, for an MPO $\vrho = [\mX_1,\dots,\mX_n]$ with MPO ranks $2\vr$, we consider the set
\begin{eqnarray}
\label{set T of MPOs}
\setT_{2\vr} := \Big\{ \vrho = [\mX_1,\dots,\mX_n]:\ \|L(\mX_\ell)\|\leq 1,\ \ell\in[n-1],\ \|L(\mX_n)\|_F\leq 1\Big\},
\end{eqnarray}
where $L(\mX_\ell)$ denotes the left unfolding of the $\ell$-th factor as in \cite[Appendix B]{qin2024quantum}. Since every $\vrho\in\ol\setX_{2\vr}$ can be written in the left-canonical form, we have $\ol\setX_{2\vr}\subset\setT_{2\vr}$. Since the elements of $\setT_{2\vr}$ may not be Hermitian, we define $Z:= \sup_{\vrho\in \setT_{2\vr}} \operatorname{Re}\sum_{k=1}^{K} \eta_{k}\<\mA_k, \vrho \>$. Then we have
\begin{eqnarray}
\label{enlarged supremum}
\max_{\vrho\in \ol \setX_{2\vr}} \sum_{k=1}^{K} \eta_{k}\<\mA_k, \vrho \> = \max_{\vrho\in \ol \setX_{2\vr}} \operatorname{Re}\sum_{k=1}^{K} \eta_{k}\<\mA_k, \vrho \> \leq \sup_{\vrho\in \setT_{2\vr}} \operatorname{Re}\sum_{k=1}^{K} \eta_{k}\<\mA_k, \vrho \> = Z,
\end{eqnarray}
where the equality holds since both $\eta_k$ and $\<\mA_k, \vrho \> = \trace(\mA_k\vrho)$ are real for Hermitian $\mA_k$ and $\vrho$, and the inequality follows from $\ol\setX_{2\vr}\subset\setT_{2\vr}$. We bound $Z$ in the following four steps.

\paragraph*{Step 1: Reduction to a finite net} \ \ Following \cite[Appendix B-A]{qin2024quantum}, for each $\ell\in[n-1]$ we construct a $\xi$-net $\{L(\mX_\ell^{(1)}),\dots,L(\mX_\ell^{(N_\ell)})\}$ for the set $\{L(\mX_\ell): \|L(\mX_\ell)\|\leq 1\}$ with respect to the spectral norm, and a $\xi$-net $\{L(\mX_n^{(1)}),\dots,L(\mX_n^{(N_n)})\}$ for the set $\{L(\mX_n): \|L(\mX_n)\|_F\leq 1\}$ with respect to the Frobenius norm, where all net elements are taken inside the corresponding sets. Since the factors are complex-valued, the real dimension of the $\ell$-th factor is $2d^2(2r_{\ell-1})(2r_\ell)$, and hence the covering numbers satisfy $N_\ell\leq(\frac{4+\xi}{\xi})^{2d^2(2r_{\ell-1})(2r_\ell)}$ \cite[eqs. (87)--(88)]{qin2024quantum}. Let $\setN := \{[\mX_1^{(p_1)},\dots,\mX_n^{(p_n)}]: p_\ell\in[N_\ell],\ \ell\in[n]\}\subset\setT_{2\vr}$. Taking $\xi = \frac{1}{2n}$, we obtain
\begin{eqnarray}
\label{covering number of T}
\log|\setN| \leq \sum_{\ell=1}^{n}8d^2r_{\ell-1}r_\ell\log(8n+1)\leq C n d^2\ol r^2\log n
\end{eqnarray}
for a universal constant $C$.

Now fix any $\vrho = [\mX_1,\dots,\mX_n]\in\setT_{2\vr}$. By the construction of the nets, there exists $\vrho^{(p)} = [\mX_1^{(p_1)},\dots,\mX_n^{(p_n)}]\in\setN$ such that $\|L(\mX_\ell)-L(\mX_\ell^{(p_\ell)})\|\leq\xi$ for $\ell\in[n-1]$ and $\|L(\mX_n)-L(\mX_n^{(p_n)})\|_F\leq\xi$. By \Cref{EXPANSION_A1TOAN-B1TOBN_1}, we can write
\begin{eqnarray}
\label{telescoping decomposition}
\vrho - \vrho^{(p)} = \sum_{\ell=1}^{n}\mD_\ell, \qquad \mD_\ell := [\mX_1^{(p_1)},\dots,\mX_{\ell-1}^{(p_{\ell-1})},\mX_\ell-\mX_\ell^{(p_\ell)},\mX_{\ell+1},\dots,\mX_n].
\end{eqnarray}
For each $\ell$, all factors of $\mD_\ell$ other than the $\ell$-th one belong to the corresponding sets in \eqref{set T of MPOs}, and the $\ell$-th factor satisfies $\|L(\mX_\ell-\mX_\ell^{(p_\ell)})/\xi\|\leq1$ for $\ell\in[n-1]$ and $\|L(\mX_n-\mX_n^{(p_n)})/\xi\|_F\leq1$ for $\ell=n$. Hence, $\mD_\ell/\xi\in\setT_{2\vr}$ for all $\ell\in[n]$. Since $\vrho\mapsto\operatorname{Re}\sum_{k} \eta_{k}\<\mA_k, \vrho \>$ is real-linear, we obtain
\begin{eqnarray}
\label{net reduction}
\operatorname{Re}\sum_{k=1}^{K} \eta_{k}\<\mA_k, \vrho \> &\!\!\!\!=\!\!\!\!& \operatorname{Re}\sum_{k=1}^{K} \eta_{k}\<\mA_k, \vrho^{(p)} \> + \xi\sum_{\ell=1}^{n}\operatorname{Re}\sum_{k=1}^{K} \eta_{k}\<\mA_k, \mD_\ell/\xi \>\nonumber\\
&\!\!\!\!\leq\!\!\!\!& \max_{\vrho'\in\setN}\operatorname{Re}\sum_{k=1}^{K} \eta_{k}\<\mA_k, \vrho' \> + n\xi Z,
\end{eqnarray}
where the inequality uses $\vrho^{(p)}\in\setN$ and $\mD_\ell/\xi\in\setT_{2\vr}$. Taking the supremum over $\vrho\in\setT_{2\vr}$ on the left-hand side of \eqref{net reduction} gives $Z\leq\max_{\vrho\in\setN}\operatorname{Re}\sum_{k} \eta_{k}\<\mA_k, \vrho \> + n\xi Z$. Moreover, $Z<\infty$ since every $\vrho\in\setT_{2\vr}$ satisfies $\|\vrho\|_F\leq\|L(\mX_n)\|_F\prod_{\ell=1}^{n-1}\|L(\mX_\ell)\|\leq 1$. Therefore, using $n\xi = \frac{1}{2}$, we conclude
\begin{eqnarray}
\label{Z bounded by net}
Z\leq 2\max_{\vrho\in\setN}\operatorname{Re}\sum_{k=1}^{K} \eta_{k}\<\mA_k, \vrho \>.
\end{eqnarray}

\paragraph*{Step 2: Hermitian and traceless reduction} \ \ The pointwise concentration in Step 3 relies on Lemma~\ref{l2 norm of approximate 2-designs RIP} and \Cref{Refined multinomial concentration psi1}, which require the coefficients $\<\mA_k,\cdot\>$ to be real and the matrix to be Hermitian, traceless, and normalized. Since the net elements in $\setN$ do not necessarily satisfy these properties, we first reduce each $\vrho\in\setT_{2\vr}$ to such a matrix in two steps.

First, we take the Hermitian part $\mH:=\frac{\vrho+\vrho^\dagger}{2}$. Since $\mA_k$ is Hermitian and $\eta_k$ is real, we have
\begin{eqnarray}
\label{Hermitian part reduction}
\operatorname{Re}\sum_{k=1}^{K} \eta_{k}\<\mA_k, \vrho \> = \sum_{k=1}^{K} \eta_{k}\operatorname{Re}\<\mA_k, \vrho \> = \sum_{k=1}^{K} \eta_{k}\<\mA_k, \mH \>.
\end{eqnarray}

Second, we remove the trace component. Since $\sum_{k=1}^K \wh p_k = \sum_{k=1}^K p_k = 1$, we have $\sum_{k=1}^K\eta_k = 0$. Together with $\<\mA_k,\mId\> = \trace(\mA_k) = \frac{d^n}{K}$ for all $k\in[K]$, this implies $\sum_{k=1}^K\eta_k\<\mA_k,\mId\> = \frac{d^n}{K}\sum_{k=1}^K\eta_k = 0$, i.e., adding any multiple of $\mId$ does not change the sum. Therefore,
\begin{eqnarray}
\label{removing trace component}
\sum_{k=1}^{K} \eta_{k}\<\mA_k, \mH \> = \sum_{k=1}^{K} \eta_{k}\<\mA_k, \mH \> - \frac{\trace(\mH)}{d^n}\sum_{k=1}^{K} \eta_{k}\<\mA_k, \mId \> = \sum_{k=1}^{K} \eta_{k}\<\mA_k, \vrho_0 \>, \qquad \vrho_0:=\mH - \frac{\trace(\mH)}{d^n}\mId.
\end{eqnarray}

Combining \eqref{Hermitian part reduction} and \eqref{removing trace component}, we have $\operatorname{Re}\sum_{k} \eta_{k}\<\mA_k, \vrho \> = \sum_{k} \eta_{k}\<\mA_k, \vrho_0 \>$ for any $\vrho\in\setT_{2\vr}$, where $\vrho_0$ is Hermitian and traceless. Moreover, since $\vrho_0$ is the orthogonal projection of $\mH$ onto the set of traceless matrices, $\|\vrho_0\|_F\leq\|\mH\|_F\leq\frac{1}{2}(\|\vrho\|_F+\|\vrho^\dagger\|_F)=\|\vrho\|_F\leq 1$. In addition, since $\vrho$ and $\vrho^\dagger$ have MPO ranks at most $2\vr$ and $\mId$ is an MPO with all MPO ranks equal to one, $\vrho_0$ has MPO ranks at most $(4r_1+1,\dots,4r_{n-1}+1)$. Therefore, whenever $\vrho_0\neq\vzero$, the normalized matrix $\wt\vrho := \vrho_0/\|\vrho_0\|_F$ belongs to $\ol\setX_{4\vr+1}$, i.e., the set in \eqref{SetOfMPO normalized} with MPO ranks $(4r_1+1,\dots,4r_{n-1}+1)$.

\paragraph*{Step 3: Pointwise concentration} \ \ Fix $\vrho\in\setN$ and let $\vrho_0$ be defined in \eqref{removing trace component}. If $\vrho_0 = \vzero$, then $\operatorname{Re}\sum_{k}\eta_k\<\mA_k,\vrho\> = 0$. Otherwise, let $a_k = \<\mA_k, \wt\vrho\>$. By the definition of $\beta$ in \eqref{requirement of t designs in the main paper 12}, $|a_k|\leq \frac{\beta}{K}$ for all $k$. Moreover, since $\wt\vrho$ is Hermitian and traceless with $\|\wt\vrho\|_F = 1$, Lemma~\ref{l2 norm of approximate 2-designs RIP} gives $\sum_{k=1}^K a_k^2 = \frac{d^n}{K(d^n+1)}$, and hence $S:=\sum_{k = 1}^{K}a_k^2p_{k}$ satisfies
\begin{eqnarray}
\label{two sided bound of S 2 designs}
\frac{\gamma_{\min}(\vrho^\star)d^n}{K^2(d^n+1)}\leq S \leq \frac{\gamma_{\max}(\vrho^\star)d^n}{K^2(d^n+1)}.
\end{eqnarray}
To apply \Cref{Refined multinomial concentration psi1} with threshold $\frac{\tau}{2}$, we need to verify
\begin{eqnarray}
    \label{Psi requirement 2 design}
    \frac{\sum_{k = 1}^{K}a_k^2p_{k}
    e^{\frac{\tau (a_k)_+}{4S}}}{S} \leq 2,
\end{eqnarray}
where $(x)_+=\max\{x,0\}$. Since $(a_k)_+\leq\frac{\beta}{K}$, a sufficient condition for \eqref{Psi requirement 2 design} is $e^{\frac{\tau\beta}{4KS}}\leq 2$, which, by the lower bound in \eqref{two sided bound of S 2 designs}, holds whenever
\begin{eqnarray}
    \label{Psi requirement 2 design2}
    \tau \leq \frac{4\log 2\cdot \gamma_{\min}(\vrho^\star) d^n }{\beta K(d^n+1)}.
\end{eqnarray}
Suppose that \eqref{Psi requirement 2 design2} holds. Then, by \Cref{Refined multinomial concentration psi1} and the upper bound in \eqref{two sided bound of S 2 designs}, we obtain
\begin{eqnarray}
    \label{Concentration inequality of sum of multinomial for SIC-POVM}
    \P{\sum_{k=1}^{K} \eta_{k}a_k > \frac{\tau}{2}} \leq  e^{-\frac{M\tau^2}{16S}}\leq e^{-\frac{K^2(d^n+1)M \tau^2}{16  d^n \gamma_{\max}(\vrho^\star)}}.
\end{eqnarray}
Since $\operatorname{Re}\sum_{k} \eta_{k}\<\mA_k, \vrho \> = \sum_{k} \eta_{k}\<\mA_k, \vrho_0 \>$ by Step 2 and $\|\vrho_0\|_F\leq1$, we further have
\begin{eqnarray}
    \label{Concentration inequality of sum of multinomial for approximate 2-designs: final}
    \P{\operatorname{Re}\sum_{k=1}^{K} \eta_{k}\<\mA_k, \vrho \> > \frac{\tau}{2}} = \P{\sum_{k=1}^{K} \eta_{k}a_k > \frac{\tau}{2\|\vrho_0\|_F}} \leq \P{\sum_{k=1}^{K} \eta_{k}a_k > \frac{\tau}{2}}\leq e^{-\frac{K^2(d^n+1)M \tau^2}{16  d^n \gamma_{\max}(\vrho^\star)}}.
\end{eqnarray}

\paragraph*{Step 4: Union bound} \ \ Combining \eqref{enlarged supremum}, \eqref{Z bounded by net}, \eqref{covering number of T} and \eqref{Concentration inequality of sum of multinomial for approximate 2-designs: final}, we obtain
\begin{eqnarray}
    \label{Concentration inequality of sum of multinomial for SIC-POVM: conclusion}
    \P{\max_{\vrho\in \ol \setX_{2\vr}}\sum_{k=1}^{K} \eta_{k}\<\mA_k, \vrho \> > \tau} &\!\!\!\! \leq \!\!\!\!& \P{\max_{\vrho\in \setN}\operatorname{Re}\sum_{k=1}^{K} \eta_{k}\<\mA_k, \vrho \> > \frac{\tau}{2}}\nonumber\\
    &\!\!\!\! \leq \!\!\!\!& e^{-\frac{K^2(d^n+1)M \tau^2}{16   d^n \gamma_{\max}(\vrho^\star)} + C n d^2 \ol{r}^2\log n }.
\end{eqnarray}
By taking $\tau = \frac{c_1  d\ol{r}\sqrt{ n  d^{n} \gamma_{\max}(\vrho^\star) \log n} }{K\sqrt{(d^n +1) M}}$ with a sufficiently large constant $c_1$, the exponent in \eqref{Concentration inequality of sum of multinomial for SIC-POVM: conclusion} is at most $-c_2 n d^2 \ol{r}^2\log n$, and the condition \eqref{Psi requirement 2 design2} is satisfied provided that
\begin{eqnarray}
    \label{necessary condition for M of 2 designs}
    M \geq C_0\,  \frac{nd^2 \ol{r}^2\beta^2\gamma_{\max}(\vrho^\star)\log n}{\gamma_{\min}^2(\vrho^\star)}
\end{eqnarray}
for a constant $C_0$ depending only on $c_1$.
Consequently, we obtain
\begin{eqnarray}
    \label{Concentration inequality of sum of multinomial for SIC-POVM: conclusion1}
    \P{\max_{\vrho\in \ol \setX_{2\vr}}\sum_{k=1}^{K} \eta_{k}\<\mA_k, \vrho \> \leq  \frac{c_1  d\ol{r}\sqrt{ nd^{n} \gamma_{\max}(\vrho^\star) \log n} }{K\sqrt{(d^n +1) M}}} \geq  1 - e^{- c_2 n d^2 \ol{r}^2\log n},
\end{eqnarray}
where $c_1$ and $c_2$ are constants.

Hence, with \eqref{upper bound of entire variable_SIC_POVM}, we can obtain
{\small \begin{eqnarray} \label{upper bound of noisy cross term}
    \<  \veta, \calA(\wh{\vrho} - \vrho^\star) \>\leq \frac{c_1  d\ol{r}\sqrt{ nd^n \gamma_{\max}(\vrho^\star) \log n} }{K\sqrt{(d^n +1) M}}\|\wh{\vrho} - \vrho^\star\|_F.
\end{eqnarray}}
Combining \eqref{whrho and rho^star relationship_1} and \eqref{summary of lower bound of difference}, with probability $1 - e^{- c_2 n d^2 \ol{r}^2\log n}$, we arrive at
\begin{eqnarray} \label{Concentration inequality of sum of multinomial for approximate 2-designs: conclusion2}
    \|\wh{\vrho} - \vrho^\star\|_F\leq  \frac{c_3  d \ol{r}\sqrt{ n \gamma_{\max}(\vrho^\star) \log n} }{ \sqrt{ M}},
\end{eqnarray}
where $c_3$ is a constant.
\end{proof}

\section{Proof of Theorem 2}
\label{Proof of recovery error of t-designs POVM 3 in Appe}

Before proving the theorem, we first establish the following property. Leveraging the properties of  $3$-designs \cite{ambainis2007quantum,dall2014accessible} and \Cref{ty of 2-designs}, for $\vrho, \vrho^\star\in\setX_{\vr}$,  we can get:
\begin{eqnarray}
    \label{3 terms in the 3 designs}
    \sum_{k = 1}^{K}\<\mA_k, \vrho - \vrho^\star \>^2p_{k}
    & \!\!\!\!= \!\!\!\!& \sum_{k = 1}^{K}\<\frac{d^n}{K} \vw_k\vw_k^\dagger, \vrho - \vrho^\star  \>^2\<\frac{d^n}{K} \vw_k\vw_k^\dagger, \vrho^\star  \>\nonumber\\
    &\!\!\!\! = \!\!\!\!& \sum_{k = 1}^{K} \frac{8^n}{K^3} \trace((\vw_k\vw_k^\dagger)^{\otimes 3} (\vrho - \vrho^\star)\otimes (\vrho - \vrho^\star) \otimes \vrho^\star  )\nonumber\\
    &\!\!\!\! = \!\!\!\!&\frac{d^{3n}}{K^2}\frac{6}{(d^n+2)(d^n+1)d^n}\trace( \calP_{\text{sym}^{\otimes 3}} (\vrho - \vrho^\star) \otimes (\vrho - \vrho^\star) \otimes \vrho^\star),
\end{eqnarray}
where $\calP_{\text{sym}^{\otimes 3}}$, defined in \cite[Lemma 3]{dall2014accessible}, represents the projector onto the symmetric subspace.
Additionally, based on \cite[Lemma 7]{gross2015partial} and \cite[eq. (322)]{mele2023introduction}, we have
 \begin{eqnarray}
    \label{projection of symmetric 3design}
    &\!\!\!\!\!\!\!\!  &\trace( \calP_{\text{sym}^{\otimes 3}} (\vrho - \vrho^\star)\otimes (\vrho - \vrho^\star) \otimes \vrho^\star) \nonumber\\
    &\!\!\!\! = \!\!\!\!& \frac{1}{6}\bigg((\trace(\vrho - \vrho^\star ))^2\trace(\vrho^\star) + \trace((\vrho - \vrho^\star)^2 )\trace(\vrho^\star)\nonumber\\
    & \!\!\!\!\!\!\!\! & + 2\trace((\vrho - \vrho^\star)\vrho^\star )\trace(\vrho - \vrho^\star)  + 2\trace((\vrho - \vrho^\star)^2 \vrho^\star)\bigg)\nonumber\\
    &\!\!\!\! = \!\!\!\!& \frac{1}{6}\| \vrho - \vrho^\star\|_F^2 + \frac{1}{3}\trace((\vrho - \vrho^\star)^2 \vrho^\star) .
\end{eqnarray}
Note that this can also be directly obtained from \cite[eq. (S36)]{huang2020predicting} when $\trace(\vrho - \vrho^\star) = 0$.
Given that $\vrho - \vrho^\star$ is Hermitian, we can ensure $(\vrho - \vrho^\star)^2 = (\vrho - \vrho^\star)(\vrho - \vrho^\star)^\dagger$ is PSD. Building upon \cite[Theorem 1]{coope1994matrix} using two PSD matrices $(\vrho - \vrho^\star)^2$ and $\vrho^\star$, we further deduce $0\leq \trace((\vrho - \vrho^\star)^2 \vrho^\star)\leq \trace((\vrho - \vrho^\star)^2)\trace(\vrho^\star) = \|\vrho - \vrho^\star\|_F^2$. Ultimately, we obtain
\begin{eqnarray}
    \label{3 terms in the 3 designs 1}
    \frac{d^{3n}\| \vrho - \vrho^\star\|_F^2}{K^2(d^n+2)(d^n+1)d^n} \leq \sum_{k = 1}^{K}\<\mA_k, \vrho - \vrho^\star \>^2p_{k} \leq \frac{3d^{3n}\| \vrho - \vrho^\star\|_F^2}{K^2(d^n+2)(d^n+1)d^n}.
\end{eqnarray}

In the remainder of the proof, we write $\mDelta = \wh\vrho - \vrho^\star$. For any Hermitian matrix $\vrho$ satisfying $\trace(\vrho)=0$, define
\begin{eqnarray}
    \label{def of u in t designs}
    u_k \ =\ u_k(\vrho) \ := \ d^n\, \vw_k^\dagger \vrho\, \vw_k
    \ =\ K\<\mA_k,\vrho\>,\qquad k=1,\dots,K,
\end{eqnarray}
so that $\<\mA_k,\vrho\> = u_k(\vrho)/K$. We also abbreviate $\E_{\vp}[h(u)] := \sum_{k=1}^{K}p_k\, h(u_k)$ with $p_k = \<\mA_k,\vrho^\star\>$, and write $P_{\text{Sym}^{\otimes m}} = \frac{1}{m!}\sum_{\pi\in S_{m}}\mW_\pi$, where $\mW_\pi$ permutes the $m$ tensor factors according to $\pi$. Finally, we define
\begin{eqnarray}
    \label{def of incoherence beta}
    \beta\ :=\ K \max_{k\in[K]}\ \max_{\vrho\in\ol\setX_{4\vr+1}}\big|\<\mA_k,\vrho\>\big|
    \ =\ d^n \max_{k\in[K]}\ \max_{\vrho\in\ol\setX_{4\vr+1}}\big|\vw_k^\dagger\vrho\,\vw_k\big|
\end{eqnarray}
as the coherence between the POVM and the normalized MPO set $\ol\setX_{4\vr+1}$ defined in \eqref{SetOfMPO normalized} with MPO ranks $(4r_1+1,\dots,4r_{n-1}+1)$, which coincides with the definition in \eqref{requirement of t designs in the main paper 12}. As shown after \Cref{Recovery error of various t-designs POVM 3}, $\beta\geq\sqrt{\frac{d^n}{d^n+1}}$ always holds.

\subsection{Moments of $u_k$ under a $t$-design}

\begin{lemma}
\label{moment lemma t designs}
Let $\{\vw_k\}_{k=1}^{K}$ form a spherical $t$-design, and let
$\vrho\in\C^{d^n\times d^n}$ be Hermitian with $\trace(\vrho)=0$. Then, for
every integer $m$ satisfying $2\leq m\leq t-1$,
\begin{eqnarray}
    \label{even moment bound}
    \big|\E_{\vp}[u^m]\big|\ \leq\ (m+1)!\,\|\vrho\|_F^{\,m}.
\end{eqnarray}
Consequently, for every integer $j$ satisfying $2\leq j\leq t-2$,
\begin{eqnarray}
    \label{absolute moment bound}
    \E_{\vp}\big[|u|^{j}\big]\ \leq\ (j+2)!\,\|\vrho\|_F^{\,j}.
\end{eqnarray}
In particular, if $\|\vrho\|_F = 1$, $t\geq 3$ and $d^n\geq 2$, then
\begin{eqnarray}
    \label{second moment two sided}
    \tfrac{1}{3}\ \leq\ \E_{\vp}[u^2]\ \leq\ 3.
\end{eqnarray}
\end{lemma}

\begin{proof}
Using $p_k = \frac{d^n}{K}\vw_k^\dagger\vrho^\star\vw_k$ and \eqref{def of u in t designs}, we have
\begin{eqnarray}
    \label{moment as tensor trace}
    \E_{\vp}[u^m]
    &\!\!\!\!=\!\!\!\!& \frac{d^{n(m+1)}}{K}\sum_{k=1}^{K}
    \big(\vw_k^\dagger\vrho\,\vw_k\big)^{m}\big(\vw_k^\dagger\vrho^\star\vw_k\big)\nonumber\\
    &\!\!\!\!=\!\!\!\!& d^{n(m+1)}\,\trace\bigg[\frac{1}{K}\sum_{k=1}^{K}
    (\vw_k\vw_k^\dagger)^{\otimes(m+1)}\ (\vrho^{\otimes m}\otimes\vrho^\star)\bigg]\nonumber\\
    &\!\!\!\!=\!\!\!\!& \frac{d^{n(m+1)}}{C^{m+1}_{d^n+m}}\,
    \trace\Big[P_{\text{Sym}^{\otimes(m+1)}}\ (\vrho^{\otimes m}\otimes\vrho^\star)\Big],
\end{eqnarray}
where the last line uses the $t$-design property in \eqref{the definition of t designs} with
$s = m+1\leq t$ together with \Cref{ty of 2-designs}. For the prefactor, we can derive
\begin{eqnarray}
    \label{prefactor bound}
    \frac{d^{n(m+1)}}{C^{m+1}_{d^n+m}}
    \ =\ (m+1)!\prod_{j=0}^{m}\frac{d^n}{d^n+j}\ \leq\ (m+1)!.
\end{eqnarray}
For the trace, expanding $P_{\text{Sym}^{\otimes(m+1)}}$ into permutations gives
\begin{eqnarray}
    \label{trace expansion}
    \trace\Big[P_{\text{Sym}^{\otimes(m+1)}}\vrho^{\otimes m}\otimes\vrho^\star\Big]
    &\!\!\!\!=\!\!\!\!& \frac{1}{(m+1)!}\sum_{\pi\in S_{m+1}}
    \trace\Big[\mW_\pi\big(\vrho^{\otimes m}\otimes\vrho^\star\big)\Big] \\
    &\!\!\!\!=\!\!\!\!& \frac{1}{(m+1)!}\sum_{\pi\in S_{m+1}}
    \trace_\pi\big(\vrho^{\otimes m}\otimes\vrho^\star\big),
\end{eqnarray}
where $\trace_\pi(\cdot)$ denotes the cycle trace induced by $\pi$. Since $\trace(\vrho)=0$, any term containing a singleton cycle supported on a $\vrho$-factor vanishes. For every nonzero term, each cycle containing only $\vrho$ has length $\ell\geq2$ and satisfies $|\trace(\vrho^{\ell})|\leq\|\vrho\|^{\ell-2}\|\vrho\|_F^2\leq\|\vrho\|_F^{\ell}$, while the unique cycle containing $\vrho^\star$ satisfies $|\trace(\vrho^{s}\vrho^\star)|\leq\|\vrho\|^{s}\trace(\vrho^\star) =\|\vrho\|^{s}\leq\|\vrho\|_F^{s}$, where we used $\vrho^\star\succeq\vzero$ and $\trace(\vrho^\star) = 1$. Thus, since the total number of $\vrho$-factors is $m$, $\sum_{\pi\in S_{m+1}} \trace_\pi\big(\vrho^{\otimes m}\otimes\vrho^\star\big) \leq (m+1)! \|\vrho\|_F^{m}$. Combining this with \eqref{prefactor bound} gives
\begin{eqnarray}
    \label{even moment bound proof}
    \big|\E_{\vp}[u^m]\big|\ \leq\ (m+1)!\,\|\vrho\|_F^{\,m}.
\end{eqnarray}

Next, if $j$ is even, then by \eqref{even moment bound proof}, we have $\E_{\vp}[|u|^j] = \E_{\vp}[u^j] \leq (j+1)!\|\vrho\|_F^j\leq (j+2)!\|\vrho\|_F^j$ with $ j\leq t-2$. If $j$ is odd, then by the Cauchy--Schwarz inequality, we can derive $\E_{\vp}\big[|u|^{j}\big] \leq \big(\E_{\vp}[u^{j-1}]\big)^{\frac12}\big(\E_{\vp}[u^{j+1}]\big)^{\frac12} \leq \big(j  (j+2) \big)^{\frac12} \|\vrho\|_F^{j}   \leq (j+2)!\,\|\vrho\|_F^{j}$. Hence, for every $2\leq j\leq t-2$, we get
\begin{eqnarray}
    \label{absolute moment bound proof}
    \E_{\vp}\big[|u|^{j}\big]\ \leq\ (j+2)!\,\|\vrho\|_F^{\,j}.
\end{eqnarray}

Finally, taking $m=2$ in \eqref{moment as tensor trace} and using \eqref{projection of symmetric 3design}, together with $\trace(\vrho)=0$, gives
\begin{eqnarray}
    \label{second moment exact}
    \E_{\vp}[u^2]
    = \frac{6\,d^{2n}}{(d^n+1)(d^n+2)}
    \bigg(\frac{1}{6}\|\vrho\|_F^2+\frac{1}{3}\trace(\vrho^2\vrho^\star)\bigg).
\end{eqnarray}
Since $\vrho^2\succeq\vzero$ and $\vrho^\star\succeq\vzero$, we have
$0\leq\trace(\vrho^2\vrho^\star)\leq\trace(\vrho^2)\trace(\vrho^\star)=\|\vrho\|_F^2$
by \cite[Theorem 1]{coope1994matrix}. Using
$\frac13\leq\frac{d^{2n}}{(d^n+1)(d^n+2)}\leq 1$ for $d^n\geq 2$ and
$\|\vrho\|_F=1$ gives \eqref{second moment two sided}.
\end{proof}

\subsection{Verification of the condition for the concentration inequality in \Cref{Refined multinomial concentration psi1}}

\begin{lemma}
\label{psi bound t designs}
Let $\{\vw_k\}_{k=1}^{K}$ form a spherical $t$-design, and let $\vrho\in\ol\setX_{4\vr+1}$ be defined in \eqref{SetOfMPO normalized}. Set $a_k = \<\mA_k,\vrho\> = \frac{d^n}{K}\vw_k^\dagger\vrho \vw_k$, $p_k = \frac{d^n}{K}\vw_k^\dagger\vrho^\star\vw_k$, $S=\sum_{k}a_k^2p_k$, and let $\Psi_+$ be as defined in \eqref{def of Psi}. For $\tau>0$, define
\begin{eqnarray}
    \label{defintion of theta}
    \theta:=\frac{K\tau}{2\,\E_{\vp}[u^2]}\ \leq\ \frac{1}{1440},
\end{eqnarray}
where $ \E_{\vp}[u^2] =  \frac{d^{3n}}{K}\sum_{k=1}^{K} \big(\vw_k^\dagger\vrho\,\vw_k\big)^{2}\big(\vw_k^\dagger\vrho^\star\vw_k\big)$.

If
\begin{eqnarray}
    \label{requirement of t designs}
    t\ \geq\ 6+\frac{2\,\theta\beta}{\log\frac{1}{4\theta}},
\end{eqnarray}
where $\beta$ is defined in \eqref{def of incoherence beta}, then we have
\begin{eqnarray}
    \label{the final conclusion for Phi}
    \Psi_+\big(\frac{2S}{\tau}\big)\leq 2.
\end{eqnarray}
\end{lemma}

\begin{proof}
To apply \Cref{Refined multinomial concentration psi1}, it suffices to establish \eqref{the final conclusion for Phi}. Since $S=\E_{\vp}[u^2]/K^2$ and $a_k=u_k/K$, we have $\frac{(a_k)_+}{2S/\tau}=\frac{K\tau}{2\E_{\vp}[u^2]}(u_k)_+=\theta\,(u_k)_+$ with $(u_k)_+ = \max\{u_k, 0 \}$. Moreover, the tilted weights in \eqref{def of S and tilted measure} are given by $q_k = u_k^2p_k/\E_{\vp}[u^2]$. Setting $m_0:=t-4\geq2$ and expanding the exponential, we obtain
\begin{eqnarray}
    \label{psi expansion}
    \Psi_+\Big(\frac{2S}{\tau}\Big)
    = \E_{q}\big[e^{\theta u_+}\big]
    = 1+\underbrace{\sum_{m=1}^{m_0}\frac{\theta^m}{m!}\E_{q}[u_+^m]}_{=:P}
    +\underbrace{\sum_{m>m_0}\frac{\theta^m}{m!}\E_{q}[u_+^m]}_{=:R}.
\end{eqnarray}
By the definition of $q_k$ and \Cref{moment lemma t designs}, we have
\begin{eqnarray}
    \label{tilted moment bound}
    \E_{q}[u_+^m]\ \leq\ \frac{\E_{\vp}\big[|u|^{m+2}\big]}{\E_{\vp}[u^2]}
    \ \leq\ 3\,(m+4)!,\qquad 1\leq m\leq m_0,
\end{eqnarray}
where we used $\|\vrho\|_F=1$, the lower bound in \eqref{second moment two sided}, and \eqref{absolute moment bound} with $j=m+2\leq m_0+2=t-2$.

\begin{itemize}

\item{Bounding $P$:} Using $\sum_{m\geq 0}(m+1)(m+2)(m+3)(m+4)\theta^m=24(1-\theta)^{-5}$ for $\theta<1$, we have
\begin{eqnarray}
    \label{main part bound}
    P\ \leq\ 3\sum_{m\geq1}\frac{(m+4)!}{m!}\theta^{m}
    \ =\ 72\Big[\frac{1}{(1-\theta)^{5}}-1\Big]\ \leq\ 720\,\theta\ \leq\ \frac12,
\end{eqnarray}
where we used $(1-\theta)^{-5}-1\leq 10\theta$ for $\theta\leq \frac{1}{10}$ and $\theta\leq \frac{1}{1440}$.

\item{Bounding $R$:} Set
\begin{eqnarray}
    \label{choice of lambda 0}
    \lambda_0\ :=\ \frac{m_0+1}{4e\,\theta}.
\end{eqnarray}
For each $m>m_0$, we split the expectation according to whether $u_+$ exceeds $\lambda_0$:
\begin{eqnarray}
    \label{split of each moment}
    \E_{q}\big[u_+^m\big]
    \ =\ \E_{q}\big[u_+^m \mathbbm{1}\{u_+\leq\lambda_0\}\big]
    \ +\ \E_{q}\big[u_+^m \mathbbm{1}\{u_+>\lambda_0\}\big].
\end{eqnarray}
On the event $\{u_+\leq\lambda_0\}$, we have $u_+^m\leq\lambda_0^m$, which gives $\E_{q}[u_+^m\,\mathbbm{1}\{u_+\leq\lambda_0\}]\leq\lambda_0^m$. On the other hand, by \eqref{def of incoherence beta}, $u_+\leq |u_k|\leq\beta$ for every $k$ and $\vrho\in\ol\setX_{4\vr+1}$. Therefore, we can derive
\begin{eqnarray}
    \label{bound on the large part}
    \E_{q}\big[u_+^m\,\mathbbm{1}\{u_+>\lambda_0\}\big]
    \ \leq\ \beta^{m}\ \Q_{q}\big[u_+>\lambda_0\big],
\end{eqnarray}
where $\Q_{q}\big[u_+>\lambda_0\big]\ := \!\!\sum_{k:\,(u_k)_+>\lambda_0}\!\! q_k = \mathbb{P}_{q}(u_+>\lambda_0)$.  Combining the above bounds and summing over $m>m_0$, we have
\begin{eqnarray}
    \label{tail split}
    R\ \leq\ \underbrace{\sum_{m>m_0}\frac{(\theta\lambda_0)^m}{m!}}_{=:R_0}
    \ +\ \Q_{q}\big[u_+>\lambda_0\big]\sum_{m>m_0}\frac{(\theta\beta)^m}{m!}
    \ \leq\ R_0
    \ +\ \underbrace{\Q_{q}\big[u_+>\lambda_0\big]\, e^{\theta\beta}}_{=:R_1},
\end{eqnarray}
where the last inequality follows by  $\sum_{m\geq 0}\frac{(\theta\beta)^m}{m!}=e^{\theta\beta}$. We next derive separate bounds for $R_0$ and $R_1$.

\textbf{$(i)$ Bounding $R_0$:} Since $\theta\lambda_0=\frac{m_0+1}{4e}$, the ratio of two consecutive terms in $R_0$ satisfies
\begin{eqnarray}
    \label{consecutive terms requirement}
    \frac{(\theta\lambda_0)^{m+1}/(m+1)!}{(\theta\lambda_0)^m/m!} = \frac{\theta\lambda_0}{m+1} \leq \frac{\theta\lambda_0}{m_0+1} = \frac{1}{4e} \leq  \frac12, \qquad m>m_0.
\end{eqnarray}
Thus, the tail can be dominated by a geometric series, and we have
\begin{eqnarray}
    \label{R0 tail upper bound}
    R_0 \leq \frac{(\theta\lambda_0)^{m_0+1}}{(m_0+1)!} \sum_{j=0}^{\infty}\left(\frac12\right)^j =  2\frac{(\theta\lambda_0)^{m_0+1}}{(m_0+1)!}.
\end{eqnarray}

Using $(m_0+1)! \geq \left(\frac{m_0+1}{e}\right)^{m_0+1}$, we further obtain
\begin{eqnarray}
    \label{R0 bound}
    R_0   &\leq& 2\left( \frac{e\,\theta\lambda_0}{m_0+1} \right)^{m_0+1}  = 2\cdot 4^{-(m_0+1)} \leq \frac18.
\end{eqnarray}

\textbf{$(ii)$ Bounding $R_1$:} It remains to bound $\Q_{q}[u_+>\lambda_0]$ and show that $R_1\leq\frac14$.

Since $u_+^{m_0}$ is a nonnegative random variable under $q$ and $\{u_+>\lambda_0\}=\{u_+^{m_0}>\lambda_0^{m_0}\}$, Markov's inequality together with \eqref{tilted moment bound}, applied with $m=m_0$, gives
\begin{eqnarray}
    \label{markov step}
    \Q_{q}\big[u_+>\lambda_0\big]
    \ \leq\ \frac{\E_{q}\big[u_+^{m_0}\big]}{\lambda_0^{m_0}}
    \ \leq\ \frac{3\,(m_0+4)!}{\lambda_0^{m_0}}.
\end{eqnarray}

To further bound the right-hand side, we control the factorial and $\lambda_0^{-m_0}$ separately. First, $(m_0+4)! = (m_0+4)(m_0+3)(m_0+2)(m_0+1)\,m_0!\leq (m_0+4)^4\,m_0!$. Next, the Stirling bound $m_0!\leq e\sqrt{m_0}\,(m_0/e)^{m_0}$ and \eqref{choice of lambda 0} give $\lambda_0=\frac{m_0+1}{4e\theta}\geq \frac{m_0}{4e\theta}$.  Substituting these bounds into \eqref{markov step}, we obtain
\begin{eqnarray}
    \label{truncated mass bound}
    \Q_{q}\big[u_+>\lambda_0\big]
    &\!\!\!\!\leq\!\!\!\!& 3\,(m_0+4)^4\cdot e\sqrt{m_0}\,\frac{m_0^{m_0}}{e^{m_0}}
    \cdot\frac{(4e\theta)^{m_0}}{m_0^{m_0}}\nonumber\\
    &\!\!\!\!=\!\!\!\!& 3e\sqrt{m_0}\,(m_0+4)^4\,(4\theta)^{m_0}.
\end{eqnarray}

To ensure that $R_1\leq\frac14$, it suffices, by
\begin{eqnarray}
    \label{R1 condition raw}
    3e\sqrt{m_0}\,(m_0+4)^4\,(4\theta)^{m_0}\,e^{\theta\beta}\ \leq\ \frac14.
\end{eqnarray}
Define
\begin{eqnarray}
    \label{def of L and B}
    L:=\log\frac{1}{4\theta}\ (>0),\qquad B:=\theta\beta\ (\geq 0).
\end{eqnarray}
Then $(4\theta)^{m_0}=e^{-m_0L}$ and $e^{\theta\beta}=e^{B}$, and taking logarithms in \eqref{R1 condition raw} yields
$\log\big(3e\sqrt{m_0}(m_0+4)^4\big)-m_0L+B\leq\log\frac14$. Equivalently, we have
\begin{eqnarray}
    \label{R1 condition}
    m_0 L - B\ \geq\ \log\big(12e\sqrt{m_0}\,(m_0+4)^4\big).
\end{eqnarray}
Finally, we verify \eqref{R1 condition}. By \eqref{R1 condition},
it suffices to show that
\begin{eqnarray}
    \varphi(m_0)\geq B,
\end{eqnarray}
where
\begin{eqnarray}
    \label{def of varphi}
    \varphi(m)\ :=\ mL-\log\big(12e\sqrt{m}\,(m+4)^4\big).
\end{eqnarray}
Consequently, for $m\geq2$, $\varphi'(m)=L-\frac{1}{2m}-\frac{4}{m+4}\geq 5.88-0.25-0.67>0$, so $\varphi$ is increasing on $[2,\infty)$. Moreover, \eqref{requirement of t designs} gives $t\geq 6+\frac{2B}{L}$ and hence $m_0=t-4\geq m_*:=2+\frac{2B}{L}\geq 2$. Therefore, by the monotonicity of $\varphi$, it suffices to verify $\varphi(m_*)\geq B$. Equivalently, it suffices to show
\begin{eqnarray}
    \label{reduced scalar inequality}
    h(B)  :=  2L+B-\log\Big(12e\sqrt{m_*}\,(m_*+4)^4\Big)\ \geq\ 0,
    \qquad m_*=2+\tfrac{2B}{L},
\end{eqnarray}
where we used $m_*L-B=2L+2B-B=2L+B$.

We verify \eqref{reduced scalar inequality} for all $B\geq 0$. At $B=0$, we have $m_*=2$ and
\begin{eqnarray}
    \label{value at B equals zero}
    h(0)\ =\ 2L-\log\big(12e\sqrt{2}\cdot 6^4\big)
    \   \geq\ 0.75\ >\ 0 .
\end{eqnarray}
Next, let $s=\frac{2B}{L}$, so that $\frac{ds}{dB}=\frac{2}{L}$, $m_*=2+s$, and $m_*+4=6+s$. Then we can derive
\begin{eqnarray}
    \label{derivative of h}
    h'(B)\ =\ 1-\frac{2}{L}\bigg(\frac{1/2}{2+s}+\frac{4}{6+s}\bigg)
    \ \geq\ 1-\frac{2}{5.88}\bigg(\frac{1}{4}+\frac{2}{3}\bigg)
    \ \geq\ 1-0.32\ >\ 0 .
\end{eqnarray}
Hence, $h$ is strictly increasing on $[0,\infty)$, and therefore $h(B)\geq h(0)>0$ for all $B\geq 0$. This proves \eqref{R1 condition}, and consequently $R_1\leq\frac14$.

\end{itemize}
Combining \eqref{psi expansion}, \eqref{main part bound}, \eqref{R0 bound} and $R_1\leq\frac14$ gives $\Psi_+\Big(\frac{2S}{\tau}\Big)\leq 1+\frac12+\frac18+\frac14\leq 2$.
\end{proof}

\subsection{Proof of Theorem 2}

\begin{proof}
If $\mDelta=\vzero$, the claim is trivial. Otherwise, $\mDelta/\|\mDelta\|_F\in\ol\setX_{2\vr}$. By linearity and homogeneity,
\begin{eqnarray}
    \label{reduction to normalized set t designs}
    \<\veta,\calA(\mDelta)\>\ \leq\
    \|\mDelta\|_F\ \max_{\vrho\in\ol\setX_{2\vr}}\
    \sum_{k=1}^{K}\eta_{k}\<\mA_k,\vrho\>.
\end{eqnarray}

Fix
\begin{eqnarray}
    \label{choice of tau t designs}
    \tau\ =\ \frac{c_1\,d\,\ol r\sqrt{n\log n}}{K\sqrt{M}},
\end{eqnarray}
and suppose that
\begin{eqnarray}
    \label{conditions of theorem 2}
    M\ \geq\ C_0\, n d^2\ol r^2\log n
    \qquad\text{and}\qquad
    t\ \geq\ 6+\frac{2\,\theta\beta}{\log\frac{1}{4\theta}},
\end{eqnarray}
where $C_0$ is an absolute constant depending only on $c_1$, and $\theta$ is defined in \eqref{defintion of theta}.

We first verify the first condition of \Cref{psi bound t designs}. By Steps 1--3 in Appendix~\ref{Proof of recovery error of SIC-POVM in Appe}, every net element is reduced to a matrix in $\ol\setX_{4\vr+1}$, and hence it suffices to establish a pointwise concentration bound with threshold $\frac{\tau}{2}$ for every fixed $\vrho\in\ol\setX_{4\vr+1}$. Note that each such $\vrho$ is Hermitian and traceless with $\|\vrho\|_F=1$, and hence satisfies the assumptions of \Cref{moment lemma t designs} and \Cref{psi bound t designs}. We then apply \Cref{psi bound t designs} with $\tau$ replaced by $\frac{\tau}{2}$. By \eqref{second moment two sided} and \eqref{choice of tau t designs}, we have
\begin{eqnarray}
    \label{theta bound t designs}
    \theta=\frac{K\tau}{4\E_{\vp}[u^2]}\ \leq\ \frac{3K\tau}{4}
    \ =\ \frac{3c_1}{4}\sqrt{\frac{nd^2\ol r^2\log n}{M}}\ \leq\ \frac{1}{1440},
\end{eqnarray}
provided $C_0\geq (1080\,c_1)^2$. The second condition of \Cref{psi bound t designs} coincides with the second condition in \eqref{conditions of theorem 2}. Therefore, \Cref{psi bound t designs} applies, yielding
\begin{eqnarray}
    \label{requirement of the concentration inequality}
    \Psi_+\big(\frac{4S}{\tau}\big)\leq 2
\end{eqnarray}
for every $\vrho\in\ol\setX_{4\vr+1}$.

Consequently, applying \Cref{Refined multinomial concentration psi1} with threshold $\frac{\tau}{2}$ yields, for $\vrho\in\ol\setX_{4\vr+1}$,
\begin{eqnarray}
    \label{pointwise concentration t designs}
    \P{\sum_{k=1}^{K}\eta_k\<\mA_k,\vrho\>>\frac{\tau}{2}} \leq\ e^{-\frac{M\tau^2}{16S}}  \leq\ e^{-\frac{K^2M\tau^2}{48}},
\end{eqnarray}
where we used $S=\E_{\vp}[u^2]/K^2\leq 3/K^2$ from \eqref{second moment two sided}.

Now, by \eqref{enlarged supremum} and \eqref{Z bounded by net}, the reduction in Step 2 of Appendix~\ref{Proof of recovery error of SIC-POVM in Appe} (which maps each net element to a matrix in $\ol\setX_{4\vr+1}$ and only rescales the sum by a factor at most one), and a union bound over the net $\setN$ with $\log|\setN|\leq Cnd^2\ol r^2\log n$, we obtain
\begin{eqnarray}
    \label{union bound t designs}
    \P{\max_{\vrho\in \ol \setX_{2\vr}}\sum_{k=1}^{K} \eta_{k}\<\mA_k, \vrho \> > \tau} &\!\!\!\!\leq\!\!\!\!& \P{\max_{\vrho\in \setN}\operatorname{Re}\sum_{k=1}^{K} \eta_{k}\<\mA_k, \vrho \> > \frac{\tau}{2}}\nonumber\\
    &\!\!\!\!\leq\!\!\!\!& |\setN|\max_{\vrho\in \ol\setX_{4\vr+1}}\P{\sum_{k=1}^{K} \eta_{k}\<\mA_k, \vrho \> > \frac{\tau}{2}}\nonumber\\
    &\!\!\!\!\leq\!\!\!\!& e^{-\frac{K^2M\tau^2}{48}+Cnd^2\ol r^2\log n} = e^{-(\frac{c_1^2}{48}-C)nd^2\ol r^2\log n},
\end{eqnarray}
where the last line follows from \eqref{pointwise concentration t designs}.
Choosing $c_1$ sufficiently large, we obtain
\begin{eqnarray}
    \label{Concentration inequality of sum of multinomial for t-designs larger 3: conclusion1}
    \P{\max_{\vrho\in \ol \setX_{2\vr}}\sum_{k=1}^{K} \eta_{k}\<\mA_k, \vrho \> \leq  \frac{c_1 d \ol{r}\sqrt{ n \log n} }{K\sqrt{ M}}} \geq  1 - e^{- c_3 n d^2 \ol{r}^2\log n},
\end{eqnarray}
where  $c_3$ is a constant.

Combining \eqref{whrho and rho^star relationship_1} and \eqref{summary of lower bound of difference} with \eqref{Concentration inequality of sum of multinomial for t-designs larger 3: conclusion1}, we can conclude:
\begin{eqnarray}
\label{final conclusion of 3 designs POVM}
\|\wh{\vrho} - \vrho^\star\|_F= O\bigg( \frac{ d\ol{r}\sqrt{n  \log n} }{\sqrt{ M}}\bigg).
\end{eqnarray}

\end{proof}

\section{Proof of Theorem 3}
\label{proof: lower bound sample complexity}

We begin with real-valued MPO parameterizations, as a lower bound for complex-valued MPO parameterizations then follows naturally.

A real-valued MPO representation of a density matrix $\vrho \in \R^{d^\nqbit \times d^\nqbit}$ has entries
\begin{eqnarray}
\label{real representation of MPO}
\vrho(i_1 \cdots i_\nqbit,j_1 \cdots j_\nqbit)  =  \mX_1^{i_1,j_1} \mX_2^{i_2,j_2} \cdots \mX_\nqbit^{i_\nqbit,j_\nqbit}.
\end{eqnarray}
For simplicity, we assume $r_1 = r_2 = \cdots = r_{\nqbit-1} = \overline{r}$, and we refer to $\overline{r}$ as the bond dimension of the MPO.

In total such an MPO model comprises $d^2$ vectors of size $1 \times \overline{r}$, $(\nqbit-2)d^2$ matrices of size $\overline{r} \times \overline{r}$, and $d^2$ vectors of size $\overline{r} \times 1$. The total number of real-valued entries in these vectors/matrices equals
\begin{eqnarray}
\label{degrees of freedom of MPO}
D_{\text{parms}} := d^2(\nqbit-2)\overline{r}^2 + 2d^2\overline{r}.
\end{eqnarray}
Stacking all such parameters into a real-valued vector $\vx \in \R^{D_{\text{parms}}}$, we can write $\vrho = \vrho_{\vx}$.

We are interested in the structure of the set $\calP$ of all physical density matrices that can be generated through such a parameterization:
\begin{eqnarray}
\label{compact representation of MPO}
\calP  =  \{ \vrho_{\vx}: ~ \vx \in \R^{D_{\text{parms}}}, ~ \trace(\vrho_{\vx})=1,  ~ \vrho_{\vx} \succeq 0 \} \subset \R^{d^\nqbit \times d^\nqbit}.
\end{eqnarray}
First, we consider the structure of the set of all possible density matrices without the trace or PSD constraint:
\begin{eqnarray}
\label{compact representation of MPO simplified}
\calP'' = \{ \vrho_{\vx}: ~ \vx \in \R^{D_{\text{parms}}} \} \subset \R^{d^\nqbit \times d^\nqbit}.
\end{eqnarray}
From~\cite[Lemma 4]{holtz2012manifolds}, we know that $\calP''$ is a submanifold of $\R^{d^\nqbit \times d^\nqbit}$. Its dimension is obtained by subtracting the gauge degrees of freedom ($(\nqbit-1)\overline{r}^2$) from the raw parameter count:
\begin{eqnarray}
\label{degrees of freedom of MPO1}
\operatorname{dim}(\calP'') &\!\!\!\!=\!\!\!\!& D_{\text{parms}} - (\nqbit-1)\overline{r}^2\nonumber\\
 &\!\!\!\!=\!\!\!\!& (d^2\nqbit-2d^2-\nqbit+1)\overline{r}^2 + 2d^2\overline{r}.
\end{eqnarray}
Now consider the set of density matrices with trace $1$:
\begin{eqnarray}
\label{degrees of freedom of MPO2}
\calP' = \{ \vrho_{\vx}: ~ \vx \in \R^{D_{\text{parms}}}, ~ \trace(\vrho_{\vx})=1 \} \ \subset \calP''.
\end{eqnarray}
Since $\calP''$ is a cone (invariant under scalar multiplication), imposing the unit trace constraint restricts the set to a cross-section of the cone, reducing the dimension by exactly one. So $\calP'$ is a manifold with
\begin{eqnarray}
\operatorname{dim}(\calP') &\!\!\!\!=\!\!\!\!&\operatorname{dim}(\calP'')-1 \nonumber\\
  &\!\!\!\!=\!\!\!\!& (d^2\nqbit-2d^2-\nqbit+1)\overline{r}^2 + 2d^2\overline{r} - 1.
\label{eq:D3mpo}
\end{eqnarray}

Finally, the set of physical MPO density matrices $\calP$ is the subset of $\calP'$ that satisfies the PSD inequality constraint $\vrho_{\vx} \succeq 0$. In differential geometry, inequality constraints define a boundary or restrict the manifold to a specific region (the interior of the PSD cone) but do not reduce the local dimension of that interior. Therefore, $\calP$ is a manifold with the same dimension as $\calP'$:
\begin{eqnarray}
\label{degrees of freedom of MPO3}
\operatorname{dim}(\calP) &\!\!\!\!=\!\!\!\!& \operatorname{dim}(\calP')\nonumber\\
 &\!\!\!\!=\!\!\!\!& (d^2\nqbit-2d^2-\nqbit+1)\overline{r}^2 + 2d^2\overline{r} - 1.
\end{eqnarray}

\subsection{The Quantum Cram\'{e}r-Rao Bound}

Let $\vrho^\star \in \calP$ be a ground-truth physical MPO density matrix. We assume $\vrho^\star$ lies in the interior of $\calP$ (i.e., it is strictly positive definite) to ensure the standard differential structure and a well-defined Bures metric. Since $\calP$ is a manifold, there exists a local parameterization of $\calP$ around $\vrho^\star$. Let $\vrho_{\vtheta}$ denote such a parameterization, with $\vtheta \in \R^{\operatorname{dim}(\calP)}$. Locally, identifying $\vrho^\star$ corresponds to identifying the value of $\vtheta$ such that $\vrho^\star = \vrho_{\vtheta^\star}$.

Consider an unbiased estimator $\wh{\vtheta}$ constructed from $M$ independent measurements of the state $\vrho_{\vtheta^\star}$. The covariance matrix of this estimator, $\mSigma = \operatorname{Cov}(\wh{\vtheta})$, is bounded by the inverse of the Quantum Fisher Information Matrix (QFIM), denoted by $\calF(\vtheta^\star)$. This bound is precisely characterized by the Quantum Cram\'{e}r-Rao Bound (QCRB):
\begin{eqnarray}
\label{QCRB inequality}
    \mSigma \succeq \frac{1}{M} \calF^{-1}(\vtheta^\star).
\end{eqnarray}
Because $\vtheta \in \R^{\operatorname{dim}(\calP)}$ acts as a minimal local coordinate system for the manifold $\calP$, the parameterization is strictly non-redundant. In quantum information geometry, the QFIM acts as the Riemannian metric tensor (the Bures metric) for the state manifold. A fundamental property of Riemannian geometry is that the metric tensor, when expressed in non-redundant local coordinates at an interior point, is strictly positive definite. Consequently, $\calF(\vtheta^\star)$ is a full-rank $\operatorname{dim}(\calP) \times \operatorname{dim}(\calP)$ matrix, and its inverse is well-defined.

\subsection{Relating Covariance to Infidelity}

We measure the quality of the estimate using the infidelity $1 - F(\vrho_{\vtheta^\star}, \vrho_{\wh{\vtheta}})$, where $F(\vrho, \mSigma) = (\trace\sqrt{\sqrt{\vrho}\mSigma\sqrt{\vrho}})^2$ is the fidelity. A rigorous way to connect the expected infidelity to the covariance matrix is through the squared Bures distance, defined as $d_B^2(\vrho,\mSigma) = 2(1-\sqrt{F(\vrho,\mSigma)})$.

For a small parameter estimation error $\vdelta = \wh{\vtheta} - \vtheta^\star$, the fidelity is close to $1$.  Writing and Taylor expanding the square root ($\sqrt{1-\Delta} \approx 1 - \Delta/2$), we find that the squared Bures distance perfectly approximates the infidelity to leading order: $d_B^2 = 1 - F + O(\Delta^2)$.

It is a well-established geometric result that the Bures metric---the infinitesimal squared distance---evaluates exactly to a quadratic form weighted by the QFIM:
\begin{eqnarray}
\label{Taylor expansion of squared Bures distance}
d_B^2(\vrho_{\vtheta^\star}, \vrho_{\wh{\vtheta}}) = \frac{1}{4} \vdelta^\top \calF(\vtheta^\star) \vdelta + O(\|\vdelta\|^3).
\end{eqnarray}
By dropping the higher-order terms, we relate the expected infidelity directly to the covariance matrix $\mSigma = \operatorname{Cov}(\wh{\vtheta})$:
\begin{eqnarray}
\label{infidelity with covariance matrix}
\E[1 - F] \approx \frac{1}{4} \E[\vdelta^\top \calF(\vtheta^\star) \vdelta] = \frac{1}{4} \trace(\calF(\vtheta^\star) \mSigma).
\end{eqnarray}
Substituting the QCRB inequality $\mSigma \succeq \frac{1}{M} \calF^{-1}(\vtheta^\star)$ into this expression (which preserves the inequality since the QFIM $\calF(\vtheta^\star)$ is positive definite) yields a fundamental lower bound on the expected error:
\begin{eqnarray}
\label{infidelity with covariance matrix inequality}
\E[1 - F] \ge \frac{1}{4} \trace\bigg( \calF(\vtheta^\star) \frac{1}{M} \calF^{-1}(\vtheta^\star) \bigg).
\end{eqnarray}
The term inside the trace simplifies to $\frac{1}{M} \mathbb{I}_{\operatorname{dim}(\calP)}$, where $\mathbb{I}_{\operatorname{dim}(\calP)}$ is the identity matrix over the tangent space parameters. The trace of this identity matrix is simply the manifold dimension $\operatorname{dim}(\calP)$:
\begin{eqnarray}
\label{infidelity with degrees of freedom inequality}
\E[1 - F] \ge \frac{1}{4M} \trace(\mathbb{I}_{\operatorname{dim}(\calP)}) = \frac{\operatorname{dim}(\calP)}{4M}.
\end{eqnarray}
The above inequality can be interpreted asymptotically in the following sense: for an efficient estimator constructed from $M$ independent measurements, the parameter error scales as $\|\vdelta\| = O_p(M^{-1/2})$, where the notation $O_p$ denotes a probabilistic $O$ statement. Consequently, the infidelity $1-F$ scales as $O_p(M^{-1})$. This implies that the truncation of the fidelity square root ($O(\Delta^2) \sim O_p(M^{-2})$) decays even faster than the truncation of the higher-order geometric terms ($O(\|\vdelta\|^3) \sim O_p(M^{-3/2})$). Taking expectations, the leading quadratic term scales as $O(M^{-1})$, while all truncated terms decay strictly faster. Thus, in the limit of large sample sizes, the quadratic metric strictly dominates, yielding the rigorous asymptotic bound $\E[1 - F] \ge \frac{\operatorname{dim}(\calP)}{4M} + o(M^{-1})$.

\subsection{Resulting Sample Complexity}

In the asymptotic regime of high-precision estimation where the target expected infidelity $\zeta$ is small and $M$ is sufficiently large, the $o(M^{-1})$ approximation errors vanish, and the leading-order bound strictly governs the sample complexity. To achieve $\E[1 - F] \le \zeta$, the number of samples $M$ must therefore satisfy:
\begin{eqnarray}
\label{infidelity upper bound requirement}
\frac{\operatorname{dim}(\calP)}{4M} \le \zeta \implies M \ge \frac{\operatorname{dim}(\calP)}{4\zeta}.
\end{eqnarray}
Substituting the dimension $\operatorname{dim}(\calP)$ derived in~\eqref{eq:D3mpo},
\begin{eqnarray}
\label{lower bound of sample complexity MPO}
M \ge \frac{(d^2\nqbit-2d^2-\nqbit+1)\overline{r}^2 + 2d^2\overline{r}-1}{4\zeta}.
\end{eqnarray}
In the thermodynamic limit of large system size $\nqbit$, the dimension scales as $\operatorname{dim}(\calP) \approx \nqbit(d^2-1)\overline{r}^2$. Thus, the asymptotic sample complexity scales linearly with the number of qudits and quadratically with both the physical Hilbert space dimension of a single site and the bond dimension:
\begin{eqnarray}
M = \Omega\bigg( \frac{\nqbit d^2 \overline{r}^2}{\zeta} \bigg).
\label{eq:MlowerMPO}
\end{eqnarray}
This confirms that estimating a state within the physical MPO manifold requires a number of samples proportional to the number of independent parameters in the model.

Moreover, to relate this bound on the infidelity to the operationally meaningful trace distance, we invoke the Fuchs--van de Graaf inequality:
\begin{equation}
\label{inequality between infidelity and trace norm}
1 - F(\vrho^\star,\wh{\vrho}) \le \frac{1}{2}\|\vrho^\star-\wh{\vrho}\|_1.
\end{equation}
Consequently, ensuring a trace-norm error $\|\vrho^\star - \wh{\vrho}\|_1 \le \zeta$ requires a comparable number of samples,
\begin{eqnarray}
M = \Omega\bigg( \frac{\nqbit d^2 \overline{r}^2}{\zeta} \bigg),
\label{eq:MlowerMPO trace norm}
\end{eqnarray}
demonstrating that the scaling of the sample complexity remains the same when measured in trace distance.

The above derivation focused on real-valued MPO parameterizations. For a complex-valued MPO, the sample complexity lower bound~\eqref{eq:MlowerMPO} still holds as-is, since $\calP$ is a subset of the complex-valued MPO density matrices.

\section{Proof of Theorem 4}
\label{Proof of spectral initialization in the GD}

\begin{proof}
By the definition of the restricted  Frobenius norm  \eqref{Definition of the restricted F norm1}, we can analyze
\begin{eqnarray}
    \label{proof of upper bound spec}
    \|\wh\vrho - \vrho^\star\|_F
    &\!\!\!\! = \!\!\!\!& \bigg\|\calP_{\trace}\bigg(\text{SVD}_{\vr}^{tt}\bigg( \sum_{k=1}^{K}\frac{K(d^n+1)}{d^n} (\<\mA_k, \vrho^\star \> + \eta_k )\mA_k   \bigg) - \vrho^\star \bigg) \bigg\|_{F,2\ol r}\nonumber\\
    &\!\!\!\! \leq \!\!\!\!& \bigg\|\text{SVD}_{\vr}^{tt}\bigg( \sum_{k=1}^{K}\frac{K(d^n+1)}{d^n} (\<\mA_k, \vrho^\star \> + \eta_k )\mA_k   \bigg) - \vrho^\star  \bigg\|_{F, 2\ol r}\nonumber\\
    &\!\!\!\! \leq \!\!\!\!& (1+ \sqrt{n-1}) \bigg\| \sum_{k=1}^{K}\frac{K(d^n+1)}{d^n} (\<\mA_k, \vrho^\star \> + \eta_k )\mA_k - \vrho^\star  \bigg\|_{F, 2\ol r}\nonumber\\
    &\!\!\!\! \leq \!\!\!\!& (1+ \sqrt{n-1}) \bigg(\frac{K(d^n+1)}{d^n}\bigg\|\sum_{k=1}^{K}  \eta_k \mA_k  \bigg\|_{F, 2\ol r}\nonumber\\
    &\!\!\!\!\!\!\!\!&+ \bigg\| \sum_{k=1}^{K}\frac{K(d^n+1)}{d^n}\<\mA_k, \vrho^\star \> \mA_k - \vrho^\star  \bigg\|_{F, 2\ol r}  \bigg)\nonumber\\
    & \!\!\!\!\leq\!\!\!\! & (1+ \sqrt{n-1}) \frac{c_1 d\ol r \sqrt{(d^n +1)n \gamma_t(\vrho^\star) \log n}}{ \sqrt{d^nM}}  ,
\end{eqnarray}
where the first two inequalities respectively follows the nonexpansiveness property of the projection onto the convex set and the quasi-optimality property of TT-SVD projection \cite{Oseledets11}.
In the last line, the upper bound of $\|\sum_{k=1}^{K}  \eta_k \mA_k  \|_{F, 2\ol r}$ can be directly obtained from  \eqref{Concentration inequality of sum of multinomial for SIC-POVM: conclusion} with probability $ 1 - e^{- c_2 n d^2\ol{r}^2\log n}$ with constants $c_i,i=1,2$. Additionally, based on  \Cref{l2 norm of approximate 2 designs-POVM RIP two density}, we have
\begin{eqnarray}
    \label{proof of I F norm}
     \bigg\| \sum_{k=1}^{K}\frac{K(d^n+1)}{d^n}\<\mA_k, \vrho^\star \> \mA_k - \vrho^\star  \bigg\|_{F, 2\ol r}
     &\!\!\!\! = \!\!\!\!& \max_{\wt\vrho \in \ol \setX_{2\vr}} \bigg\< \sum_{k=1}^{K}\frac{K(d^n+1)}{d^n}\<\mA_k, \vrho^\star \> \mA_k - \vrho^\star,  \wt \vrho \bigg\> \nonumber\\
     &\!\!\!\! = \!\!\!\!&\max_{\wt\vrho \in \ol \setX_{2\vr} }  \<\vrho^\star,  \wt \vrho \> + \trace(\vrho^\star)\trace(\wt \vrho) - \<\vrho^\star,  \wt \vrho \>\nonumber\\
     &\!\!\!\!=\!\!\!\!& 0.
\end{eqnarray}

\end{proof}

\section{Proof of Theorem 5}
\label{Proof of convergence rate in the GD}

\begin{proof}
For any density matrices $\vrho_1,\vrho_2\in\setX_{\vr}$ in \eqref{SetOfMPO}, we have
\begin{eqnarray}
    \label{Definition of the restricted F norm}
    \|\vrho_1 - \vrho_2\|_F^2 & = & \<\vrho_1 - \vrho_2, \vrho_1 - \vrho_2  \>\nonumber\\
    & = &\hspace{-1.5cm}  \max_{\mbox{\tiny$\begin{array}{c} \vrho\in\C^{d^n\times d^n}, \vrho = \vrho^\dagger,\\
\text{rank}(\vrho)=(2r_1,\dots,2r_{n-1}),\\
     \trace(\vrho) = 0, \|\vrho\|_F \leq \|\vrho_1 - \vrho_2\|_F \end{array}$}} \hspace{-1cm} \<\vrho_1 - \vrho_2,  \vrho  \>.
\end{eqnarray}

Furthermore, we define a restricted Frobenius norm as following:
\begin{eqnarray}
    \label{Definition of the restricted F norm1}
    \|\vrho_1 - \vrho_2\|_{F,2\ol{r}} =\max_{\vrho\in\ol \setX_{2\vr}} \<\vrho_1 - \vrho_2,  \vrho  \>.
\end{eqnarray}

Now, we can expand $\|\vrho^{(\tau+1)} - \vrho^\star\|_F^2$ as following:
\begin{eqnarray}
    \label{expansion of vrho - rvrho*}
    & \!\!\!\!\!\!\!\! &\|\vrho^{(\tau+1)} - \vrho^\star\|_F^2 = \|\calP_{\setX_{\vr}}(\vrho^{(\tau)} - \mu \nabla g(\vrho^{(\tau)})  ) - \vrho^\star \|_{F,2\ol r}^2 \nonumber\\
    &\!\!\!\! \leq \!\!\!\! & (1+ b )\|\vrho^{(\tau)} - \mu \nabla g(\vrho^{(\tau)})   - \vrho^\star \|_{F,2\ol r}^2\nonumber\\
    &\!\!\!\! = \!\!\!\! & (1+ b ) \bigg(\|\vrho^{(\tau)} - \vrho^\star\|_F^2 - 2\mu\<\vrho^{(\tau)} - \vrho^\star, \nabla g(\vrho^{(\tau)}) \> + \mu^2\|\nabla g(\vrho^{(\tau)}) \|_{F,2\ol r}^2\bigg).
\end{eqnarray}

First, we need to analyze
\begin{eqnarray}
    \label{cross term in the GD TT}
    & \!\!\!\!\!\!\!\! &\<\vrho^{(\tau)} - \vrho^\star, \nabla g(\vrho^{(\tau)}) \> \nonumber\\
    &\!\!\!\! = \!\!\!\!& \sum_{k=1}^{K}\<\mA_k, \vrho^{(\tau)}- \vrho^\star  \>^2 - \sum_{k=1}^{K}\eta_k\<\mA_k, \vrho^{(\tau)}- \vrho^\star  \>\nonumber\\
    &\!\!\!\! \geq \!\!\!\!&\frac{d^n\| \vrho^{(\tau)}- \vrho^\star\|_F^2}{K(d^n+1)} - \frac{c_1  d\ol{r}\sqrt{nd^n  \gamma_t(\vrho^\star) \log n} }{K\sqrt{(d^n +1) M}}\|\vrho^{(\tau)}- \vrho^\star\|_F,
\end{eqnarray}
where the first inequality uses Lemma~\ref{l2 norm of approximate 2-designs RIP} and \eqref{Concentration inequality of sum of multinomial for approximate 2-designs: conclusion2} with probability $ 1 - e^{- c_2 nd^2 \ol{r}^2\log n}$. Here $c_1,c_2$ are positive constants.

Then we need to derive
\begin{eqnarray}
    \label{squared term in the GD TT}
    \|\nabla g(\vrho^{(\tau)}) \|_{F,\ol 2r}^2&\!\!\!\! \leq \!\!\!\!& 2 \bigg\|\sum_{k=1}^{K}\<\mA_k, \vrho^{(\tau)}- \vrho^\star  \>\mA_k  \bigg\|_{F,2\ol r}^2 + 2 \bigg\|\sum_{k=1}^{K}\eta_k\mA_k \bigg\|_{F,2\ol r}^2\nonumber\\
    &\!\!\!\! \leq \!\!\!\!&\frac{d^{2n+1} \|\vrho^{(\tau)}- \vrho^\star\|_F^2}{K^2(d^n + 1)^2}  + \frac{2c_1^2    \ol{r}^2nd^{n+2} \gamma_t(\vrho^\star) \log n }{K^2(d^n +1) M},
\end{eqnarray}
where the second inequality respectively follows \eqref{Concentration inequality of sum of multinomial for approximate 2-designs: conclusion2} and
\begin{eqnarray}
    \label{one term in squared term in the GD TT}
    \bigg\|\sum_{k=1}^{K}\<\mA_k, \vrho^{(\tau)}- \vrho^\star  \>\mA_k  \bigg\|_{F,2\ol r}
    &\!\!\!\! =\!\!\!\! & \max_{\wt \vrho\in \ol\setX_{2 \vr}} \bigg\<\sum_{k=1}^{K}\<\mA_k, \vrho^{(\tau)}- \vrho^\star  \>\mA_k,  \wt \vrho  \bigg\>\nonumber\\
    &\!\!\!\! =\!\!\!\! &\max_{\wt \vrho\in \ol\setX_{2 \vr}} \sum_{k=1}^{K}\<\mA_k, \vrho^{(\tau)}- \vrho^\star  \>\<\mA_k,  \wt \vrho  \>\nonumber\\
    &\!\!\!\!  = \!\!\!\! & \max_{\wt \vrho\in \ol\setX_{2 \vr}}\frac{d^n  (\<\vrho^{(\tau)}- \vrho^\star, \wt \vrho \> + \trace(\vrho^{(\tau)}- \vrho^\star) \trace(\wt \vrho))}{K(d^n + 1)}\nonumber\\
&\!\!\!\! \leq\!\!\!\! & \frac{d^n \|\vrho^{(\tau)}- \vrho^\star\|_F}{K(d^n + 1)},
\end{eqnarray}
where the first inequality uses \Cref{l2 norm of approximate 2 designs-POVM RIP two density}.

Combining \eqref{cross term in the GD TT} and \eqref{squared term in the GD TT}, we arrive at
\begin{eqnarray}
    \label{expansion of vrho - rvrho* final}
    &\!\!\!\! \!\!\!\!&\|\vrho^{(\tau+1)} - \vrho^\star\|_F^2 \nonumber\\
    &\!\!\!\! \leq \!\!\!\!& (1+ b)\big( \big(1 - \frac{d^{n+1} \mu}{K(d^n+1)} + \frac{d^{2n+1} \mu^2}{K^2(d^n+1)^2}\big) \|\vrho^{(\tau)} - \vrho^\star  \|_F^2 +e \big) \nonumber\\
    &\!\!\!\! \leq\!\!\!\! &(1+ b)\big( \big(1 - \frac{d^{n} \mu}{K(d^n+1)}\big)\|\vrho^{(\tau)} - \vrho^\star  \|_F^2 +e \big) \nonumber\\
    &\!\!\!\! \leq\!\!\!\! & a^{\tau+1}\|\vrho^{(0)} - \vrho^\star  \|_F^2 + \frac{1+ b}{1-a}c_3\bigg( \frac{\mu d\ol{r}\sqrt{n \gamma_t(\vrho^\star) \log n} }{\sqrt{M} }\|\vrho^{(\tau)}- \vrho^\star\|_F  + \frac{\mu^2 \ol{r}^2nd^2 \gamma_t(\vrho^\star) \log n }{M}  \bigg) \nonumber\\
    &\!\!\!\! \leq\!\!\!\! & a^{\tau+1}\|\vrho^{(0)} - \vrho^\star  \|_F^2 + c_4\bigg( \frac{d\ol{r}\sqrt{n  \gamma_t(\vrho^\star) \log n} }{\sqrt{M} }\|\vrho^{(\tau)}- \vrho^\star\|_F + \frac{ \ol{r}^2n d^2 \gamma_t(\vrho^\star) \log n }{M}  \bigg),
\end{eqnarray}
where we define the error term
\begin{eqnarray}
e  = \frac{2\mu c_1 d\ol{r}\sqrt{n d^n \gamma_t(\vrho^\star) \log n} }{K\sqrt{(d^n+1)M}}\|\vrho^{(\tau)}- \vrho^\star\|_F   + \frac{2 \mu^2 c_1^2 \ol{r}^2n d^{n+2}\gamma_t(\vrho^\star) \log n }{K^2(d^n+1)M},
\end{eqnarray}
and $c_3,c_4$ are positive constants. The last two lines follow  assumptions with regard to the step size and initialization: $\frac{bK(d^n+1)}{d^{n}(1+b) }<\mu\leq \frac{(d-1)(d^n+1) K}{  d^{n+1}}$ and $b< d-1 $ which guarantee that the convergence rate satisfies
\begin{eqnarray}
    \label{requirement of converence rate}
    a = (1+ b ) \bigg(1 - \frac{d^{n} \mu}{K(d^n+1)} \bigg)  <1.
\end{eqnarray}
According to \eqref{expansion of vrho - rvrho* final}, the linear convergence of the PGD algorithm is formally established. Furthermore, when the number of iterations is sufficiently large such that $a^{\tau+1}\|\vrho^{(0)} - \vrho^\star  \|_F^2\leq c_4  \frac{ \ol{r}^2n d^2 \gamma_t(\vrho^\star) \log n }{M} $, the recovery error of the PGD algorithm satisfies:
\begin{eqnarray}
    \label{recovery error of PGD in the appendix}
    \|\vrho^{(\tau)} - \vrho^\star\|_F \leq c_5 \sqrt{\frac{n d^2\ol{r}^2 \gamma_t(\vrho^\star) \log n }{M }},
\end{eqnarray}
where $c_5$ is a positive constant. This completes the proof of Theorem 5.
\end{proof}

\section{Auxiliary Materials}
\label{sec:app-aux}

\begin{lemma} (\cite[Lemma 3]{dall2014accessible})
\label{ty of 2-designs}
The integral over the Haar measure in \eqref{the definition of t designs} is given by
\begin{eqnarray} \label{the definition of integral in t designs}
     \int(\vw\vw^\dagger)^{\otimes s} d\vw = \frac{1}{C_{d^n+s-1}^s} P_{\text{Sym}},
\end{eqnarray}
where $P_{\text{Sym}}$ is the projector over the symmetric subspace.
\end{lemma}

\begin{lemma}(\cite[Lemma 10]{qin2024quantum})
\label{EXPANSION_A1TOAN-B1TOBN_1}
For any  ${\bm A}_i,{\bm A}^\star_i\in\R^{r_{i-1}\times r_i},i=1,\dots,N$, we have
\begin{eqnarray}
    \label{EXPANSION_A1TOAN-B1TOBN_2}
    {\bm A}_1{\bm A}_2\cdots {\bm A}_N-{\bm A}_1^\star{\bm A}_2^\star\cdots {\bm A}_N^\star =  \sum_{i=1}^N \mA_1^\star \cdots \mA_{i-1}^\star (\mA_{i} - \mA_i^\star) \mA_{i+1} \cdots \mA_N.
\end{eqnarray}

\end{lemma}

\begin{lemma}
\label{l2 norm of approximate 2 designs-POVM RIP two density}
For arbitrary Hermitian matrix $\vrho_1, \vrho_2\in\C^{d^n\times d^n}$ and a set of  $2$-design POVMs $\mA_k = \frac{d^n}{K} \vw_k\vw_k^\dagger, k\in[K] $ in Definition~\ref{definition_of_T_Design}, we have
\begin{eqnarray}
\label{The l2 norm of A(rho) two densityapproximate 2 designs}
\sum_{k=1}^K\<\mA_k,\vrho_1 \>\<\mA_k, \vrho_2\> =\frac{d^n ( \trace(\vrho_1)\trace(\vrho_2) +\trace(\vrho_1 \vrho_2) )}{K(d^n+1)}.
\end{eqnarray}
\end{lemma}
\begin{proof}
\begin{eqnarray}
\label{The square of rho two density  2 designs}
\sum_{k=1}^K\<\mA_k,\vrho_1 \>\<\mA_k, \vrho_2\>
&\!\!\!\! = \!\!\!\!&\sum_{k=1}^K\trace((\mA_k\vrho_1)\otimes (\mA_k\vrho_2))\nonumber\\
&\!\!\!\! =\!\!\!\! & \sum_{k=1}^K\trace(\mA_k^{\otimes 2}(\vrho_1\otimes \vrho_2) )\nonumber\\
&\!\!\!\! =\!\!\!\! & \frac{d^n \trace( (\mId + \mS) (\vrho_1\otimes \vrho_2))}{K(d^n+1)}\nonumber\\
&\!\!\!\! = \!\!\!\!& \frac{d^n ( \trace(\vrho_1)\trace(\vrho_2) +\trace(\vrho_1 \vrho_2) )}{K(d^n+1)},
\end{eqnarray}
where the penultimate line follows \eqref{the definition of t designs} (with $s=2$) and \eqref{sum of 2-designs measurement}, and the last line is from \cite[Lemma 17]{KuengACHA17}.
\end{proof}

The concentration bound in \cite[Lemma 3]{qin2024quantum} controls the statistical error through the largest coefficient $a_{\max}=\max_{i,k}|a_{i,k}|$, which enters as a hard constraint on the admissible range of the Chernoff parameter and thus forces a case distinction between the sub-Gaussian and the sub-exponential regimes. Below we establish a refined version in which $a_{\max}$ is replaced by a weighted exponential-moment condition on the coefficients. The resulting bound consists of a single sub-Gaussian term, requires no case distinction, and treats coefficients of arbitrary sign simultaneously, so that the splitting argument of \cite[Lemma 3]{qin2024quantum} is no longer needed.

\begin{lemma}
\label{Refined multinomial concentration psi1}
Suppose that the $Q$ multivariate random variables $(f_{i,1},\dots, f_{i,K}), i=1,\dots,Q$ are mutually independent and follow the multinomial distribution $\operatorname{Multinomial}(M,\vp_i)$ with $\sum_{k=1}^{K}f_{i,k}=M$ and $\vp_i = [p_{i,1},\dots, p_{i,K}]$. Let $a_{i,1},\dots,a_{i,K}\in\R$ be fixed (of arbitrary sign) and define
\begin{eqnarray}
    \label{def of S and tilted measure}
    S \ :=\ \sum_{i=1}^{Q}\sum_{k=1}^{K} a_{i,k}^2 p_{i,k}>0,
    \qquad
    q_{i,k}\ :=\ \frac{a_{i,k}^2 p_{i,k}}{S},
\end{eqnarray}
so that $\{q_{i,k}\}_{i\in[Q],k\in[K]}$ is a probability distribution. Denote by
\begin{eqnarray}
    \label{def of Psi}
    \Psi_+(\lambda)\ :=\ \E_{q}\Big[e^{(a)_+/\lambda}\Big]
    \ =\ \frac{1}{S}\sum_{i=1}^{Q}\sum_{k=1}^{K} a_{i,k}^2 p_{i,k}\,
    e^{\frac{(a_{i,k})_+}{\lambda}},
    \qquad \lambda>0,
\end{eqnarray}
where $(x)_+=\max\{x,0\}$. Then, for any $\tau>0$ satisfying
\begin{eqnarray}
    \label{condition of psi1}
    \Psi_+\Big(\frac{2S}{\tau}\Big)\leq 2,
\end{eqnarray}
we have
\begin{eqnarray}
    \label{Refined multinomial concentration bound}
    \P{\sum_{i=1}^{Q}\sum_{k=1}^{K} a_{i,k}\Big(\frac{f_{i,k}}{M}-p_{i,k}\Big)> \tau}
    \ \leq\ e^{-\frac{M \tau^2}{4S}}.
\end{eqnarray}
\end{lemma}

\begin{proof}
First, we can derive
\begin{eqnarray}
    \label{elementary inequality}
    e^{y}\ \leq\ 1+y+\frac{y^2}{2}e^{y_+},\qquad y\in\R.
\end{eqnarray}
where $e^{y}-1-y=\int_0^y (y-u)e^{u}du\leq e^{y}\int_0^y(y-u)du=\frac{y^2}{2}e^{y}$ for $y\geq 0$, while for $y<0$ we have $e^{y}-1-y=\int_y^0 (u-y)e^{u}du\leq \int_y^0(u-y)du=\frac{y^2}{2}$.

Next, we fix $v>0$ and set $y_{i,k}=\frac{v a_{i,k}}{M}$, so that
$(y_{i,k})_+=\frac{v (a_{i,k})_+}{M}$. Since $(f_{i,1},\dots,f_{i,K})$ follows a
multinomial distribution, its moment generating function admits the exact
expression
\begin{eqnarray}
    \label{exact MGF}
    \E{e^{v\sum_{k=1}^{K}a_{i,k}\frac{f_{i,k}}{M}}}
    \ =\ \E\Big[\prod_{k=1}^{K}e^{y_{i,k}f_{i,k}}\Big]
    \ =\ \Big(\sum_{k=1}^{K}p_{i,k}e^{y_{i,k}}\Big)^{M}.
\end{eqnarray}
Applying \eqref{elementary inequality} to each $y_{i,k}$ and using
$1+u\leq e^{u}$,
\begin{eqnarray}
    \label{MGF single POVM}
    \Big(\sum_{k=1}^{K}p_{i,k}e^{y_{i,k}}\Big)^{M}
    \ \leq\
    e^{v\sum_{k=1}^{K}a_{i,k}p_{i,k}
    +\frac{v^2}{2M}\sum_{k=1}^{K}a_{i,k}^2p_{i,k}e^{\frac{v(a_{i,k})_+}{M}}}.
\end{eqnarray}
By the mutual independence of $(f_{i,1},\dots,f_{i,K}), i=1,\dots,Q$, taking the
product over $i$ and recalling \eqref{def of S and tilted measure} and
\eqref{def of Psi} with $\lambda=\frac{M}{v}$ yields
\begin{eqnarray}
    \label{MGF all POVMs}
    \E{e^{v\sum_{i=1}^{Q}\sum_{k=1}^{K}a_{i,k}\frac{f_{i,k}}{M}}}
    \ \leq\ e^{v\sum_{i,k}a_{i,k}p_{i,k}
    +\frac{v^2 S}{2M}\Psi_+\big(\frac{M}{v}\big)}.
\end{eqnarray}

By Markov's inequality and \eqref{MGF all POVMs}, we have
\begin{eqnarray}
    \label{Chernoff refined}
    \P{\sum_{i,k} a_{i,k}\Big(\frac{f_{i,k}}{M}-p_{i,k}\Big)> \tau}
    \ \leq\ e^{-v\big(\tau +\sum_{i,k}a_{i,k}p_{i,k}\big)}
    \E{e^{v\sum_{i,k}a_{i,k}\frac{f_{i,k}}{M}}}
    \ \leq\ e^{-v\tau+\frac{v^2S}{2M}\Psi_+\big(\frac{M}{v}\big)}.
\end{eqnarray}
Choosing $v=\frac{M\tau}{2S}$, for which $\frac{M}{v}=\frac{2S}{\tau}$, and invoking
the assumption \eqref{condition of psi1}, the exponent in \eqref{Chernoff refined}
is bounded by
\begin{eqnarray}
    \label{final exponent}
    -\frac{M \tau^2}{2S}+\frac{M \tau^2}{8S}\,\Psi_+\Big(\frac{2S}{\tau}\Big)
    \ \leq\ -\frac{M\tau^2}{2S}+\frac{M\tau^2}{4S}\ =\ -\frac{M\tau^2}{4S},
\end{eqnarray}
which completes the proof.
\end{proof}

As a direct result of \cite[Theorem 1.3]{vershynin2020concentration}, we have
\begin{lemma}
\label{lemma: concentration of tensor product}
Let $n$ and $d$ be positive integers
and $f: \C^{d^n}\to \R$ be a convex and Lipschitz function. Consider a vector $\vx = \vx_1\otimes \cdots \otimes \vx_n\in\C^{d^n}$ where all $\vx_k\in\C^{d}$ are independent random vectors on the unit sphere. Then, for every $0\leq t \leq 2(\E|f(\vx)|^2)^{\frac{1}{2}}$, we have
\begin{eqnarray}
    \label{concentration inequality for tensor product}
    \P{|f(\vx) - \E f(\vx)|\geq t}\leq 2e^{-\frac{(d-1)t^2}{n L^2}},
\end{eqnarray}
where $L$ is the Lipschitz constant of $f(\vx)$.
\end{lemma}

\end{document}